\documentclass[twocolumn,journal]{IEEEtran}
\IEEEoverridecommandlockouts
\usepackage[applemac]{inputenc}
\usepackage{graphics} % for pdf, bitmapped graphics files
\usepackage{epsfig} % for postscript graphics file
\usepackage{epstopdf}
\usepackage{algorithm}
\usepackage[noend]{algpseudocode}
\usepackage{pgfplots} 

\usepackage{tikz}
\usetikzlibrary{automata,positioning,shapes,arrows}

\usepackage{amsmath} % assumes amsmath package installed
\usepackage{amssymb}  % assumes amsmath package installed
\usepackage{amsthm}
\usepackage{epstopdf}
\usepackage{graphicx,xcolor}
\usepackage{caption}
\usepackage{subcaption}
\usepackage{courier}
\usepackage{mathrsfs}
\usepackage{tikz}
\usetikzlibrary{calc,patterns,angles,quotes}
\usepackage{todonotes}
\makeatletter
\def\BState{\State\hskip-\ALG@thistlm}
\makeatother

\newtheorem{thm}{Theorem}
\newtheorem{defi}{Definition}

\newtheorem{cor}{Corollary}

\newtheorem{prop}{Proposition}
\newtheorem{problem}{Problem}
\newtheorem{example}{Example}

\title{\LARGE \bf
Control Theory Meets POMDPs: \\A Hybrid Systems Approach
}

\author{Mohamadreza Ahmadi, Nils Jansen, Bo Wu, and Ufuk Topcu
\thanks{This work was supported by AFOSR FA9550-19-1-0005, DARPA D19AP00004, NSF 1646522 and NSF 1652113.}
\thanks{M. Ahmadi is with the Center for Autonomous Systems and Technologies at the California Institute of Technology, 1200 E. Calif. Blvd., MC 104-44, Pasadena, CA 91125.  e-mail: (\{mrahmadi\}@caltech.edu). B. Wu and U. Topcu are with the Oden Institute for Computational Engineering and Sciences, University of Texas, Austin, 201 E 24th St, Austin, TX 78712, USA.  e-mail: (\{ bwu3,  utopcu\}@utexas.edu). N. Jansen is with the  Department of Software Science (SWS) at the Radboud University, Institute for Computing and Information Sciences, PO Box 9010, 6500 GL Nijmegen, Netherlands. e-mail: (\{n.jansen\}@science.ru.nl).
}% <-this % stops a space
}

\begin{document}

\maketitle
\thispagestyle{empty}
\pagestyle{empty}

%%%%%%%%%%%%%%%%%%%%%%%%%%%%%%%%%%%%%%%%%%%%%%%%%%%%%%%%%%%%%%%%%%%%%%%%%%%%%%%%
\begin{abstract}
Partially observable Markov decision processes (POMDPs)  provide a modeling framework for a variety of sequential decision making under uncertainty scenarios in artificial intelligence (AI). Since the states are not directly observable in a POMDP, decision making has to be performed based on the output of a Bayesian filter (continuous beliefs). Hence, POMDPs are often computationally intractable to solve exactly and researchers resort to approximate methods often using discretizations of the continuous belief space. These approximate solutions are, however, prone to discretization errors, which has made POMDPs ineffective in applications, wherein guarantees for safety, optimality, or performance are required. To overcome the complexity challenge of POMDPs, we apply notions from control theory. 
The goal is to determine the reachable belief space of a POMDP, that is, the set of all possible evolutions given an initial belief distribution over the states and a set of actions and observations. 
We begin by casting the problem of analyzing a POMDP into analyzing the behavior of a discrete-time switched system.
For estimating the reachable belief space, we find over-approximations in terms of sub-level sets of Lyapunov functions. 
Furthermore, in order to verify safety and optimality requirements of a given POMDP, we formulate a barrier certificate theorem, wherein we show that if there exists a barrier certificate satisfying a set of inequalities along with the belief update equation of the POMDP, the safety and optimality properties are guaranteed to hold.  In both cases, we show how the calculations can be decomposed into smaller problems that can be solved in parallel. The conditions we formulate can be computationally implemented as a set of sum-of-squares programs.  We illustrate the applicability of our method by addressing two problems in active ad scheduling and machine teaching.
%\nils{safety, optimality, and performance are always mixed up, we should decide on what we state where (also later).}

\end{abstract}

%%%%%%%%%%%%%%%%%%%%%%%%%%%%%%%%%%%%%%%%%%%%%%%%%%%%%%%%%%%%%%%%%%%%%%%%%%%%%%%%
\section{INTRODUCTION}

%Emerging control applications, in particular concerned with safety-critical systems, require system analysis  and controller synthesis methods that are resilient to abrupt system changes. For instance, consider an unmanned aerial vehicle flying on a specified state evolution . Due to external conditions, such as a wind gust, severe damage is incurred to one of the wings~\cite{f15Israel}. The dynamics of the aircraft after the incident does not follow the equations of motions based on which the system was initially designed. Hence, to preclude further damage or \emph{safe} landing, we require data-driven methods for system analysis  and control synthesis. 

A  formal model for planning subject to stochastic behavior is a Markov decision process (MDP)~\cite{Put94}, where an agent chooses to perform an action under full knowledge of the environment it is operating in.
The outcome of an action is a probability distribution over the system states.
Many applications, however, allow only \emph{partial observability} of the current system state, say via noisy sensor outputs~\cite{kaelbling1998planning,thrun2005probabilistic,WongpiromsarnF12}.
% where at each state the system is in, an agent makes decisions upon which a successor state is probabilistically determined.
%A typical scenario describes an {agent} navigating through an environment while having to avoid a randomly-moving {opponent}.
%A state of the system describes the positions of all agents, and upon the agent's choice of direction, the robot moves into that direction and the new opponent position is randomly determined.
%While the opponent may follow unknown intentions, learning approaches sufficiently abstract behavior perceived over time into well-defined probability distributions~\cite{Russell-AI-Modern}. 
%Example applications are autonomous trading agents in the stock market~\cite{wellman2001designing} or service robots in presence of humans~\cite{khandelwal2017bwibots}.
%Examples for behavioral learning in such partially controlled multi-agent systems~\cite{brafman1996partially} are autonomous trading agents in the stock market~\cite{wellman2001designing} or service robots in presence of humans~\cite{khandelwal2017bwibots}.
Partially observable Markov decision processes (POMDPs) extend MDPs to account for such partial information~\cite{Russell-AI-Modern}.
Upon certain \emph{observations}, the agent infers the likelihood of the system being in a certain state, called the belief state. 
The belief state together with an update function form a (typically uncountably infinite) MDP, referred to as the \emph{belief MDP}~\cite{ShaniPK13,MadaniHC99,braziunas2003pomdp}.
However, solving a POMDP exactly, i.e., synthesizing a policy ensuring optimality or a certain performance, is carried out by assessing the entire belief MDP, rendering the problem undecidable~\cite{MadaniHC99,ChatterjeeCT16}. Several promising approximate point-based methods via finite abstraction of the belief space are proposed in the literature~\cite{Hauskrecht2000,wu2018privacy,Spaan2005,6284837}. At the heart of these methods is the observation that, given an initial belief $b_0$ (the initial probability of the system being in some state), the set of possible evolutions emanating from $b_0$ under arbitrary action and observation sequences,  called the \textit{reachable belief space} is much smaller than the whole belief space~\cite{hsu2007accelerating}. Hence, only a finite number of \textit{samples} or \textit{points} from the reachable belief space can be considered for decision making.%\nils{I don't believe that we have any guarantee for a finite number of samples.}

Nonetheless, computing the reachable belief space is non-trivial and fundamental in analyzing the convergence of approximate POMDP solutions. 
In fact, Pineu~\textit{et. al.}~\cite{pineau2006anytime} showed that point-based value iteration can approximate the solution of the POMDP sufficiently close, whenever  the total reward function is a discounted sum of the rewards over time, the samples are sufficiently close to each other, and if the sampling is carried out over all of the reachable belief space. 
In another vein, Lee \textit{et. al.}~\cite{lee2008makes} demonstrated that an approximate POMDP solution can be computed in time polynomial in the covering number (the number of balls covering a set) of the reachable belief space.  Yet, the problem lies in the fact that the reachable belief space is not known \textit{a priori} and therefore it is hard to make judgements on the convergence rate of a given approximate method. Several techniques have been proposed for estimating the reachable belief space mainly based on heuristics~\cite{kurniawati2008sarsop,spaan2005perseus}, "smarter" sampling methods, such as importance sampling~\cite{luo2016importance}, and graph search~\cite{agha2014robust,prentice2010belief}. However, these methods are difficult to adapt from one problem setting to another and are often \textit{ad hoc}.

Even with a reachable belief space at hand, the approximate techniques do not provide a guarantee for safety or optimality. That is, it is not clear whether the probability of satisfying a pre-specified safety/optimality requirement is an upper-bound or a lower-bound for a given POMDP. Establishing guaranteed safety and optimality is of fundamental importance in safety-critical applications, e.g. aircraft collision avoidance~\cite{Kochenderfer2013} and Mars rovers~\cite{smith2004heuristic}. %\nils{now performance means safe and optimal. Is this what we want? If yes, I'll rewrite accordingly.}

In this paper, we use notions from control theory to analyze a POMDP without the need to solve it explicitly. In this sense, our method follows the footsteps of  Lyapunov's second method for verifying the stability of all solutions of an ODE without the need to solve it or the application of barrier certificate theorems for safety verification without the need to solve an ODE for all initial conditions in a given set. Our main contributions are fourfold:
\begin{itemize}
\item We propose a switched system representation for the belief evolutions of POMDPs;
\item We present a method based on Lyapunov functions to  over-approximate the reachable belief space, thus allowing us to verify reachability objectives. These  over-approximations are in terms of sub-level sets of Lyapunov functions. We demonstrate how these calculations can be decomposed in terms of the POMDP actions;
\item We propose a method based on barrier certificates to verify optimality and safety of a given POMDP. We show how the calculations of the barrier certificates can be decomposed in terms of actions;
\item We formulate a set of sum-of-squares programs for finding the Lyapunov functions and barrier certificates to verify safety and optimality and  over-approximate the reachable belief space, respectively. 
\end{itemize}
We illustrate the efficacy of our methodology with two examples from cyber-physical systems and artificial intelligence. 

Preliminary results related to the formulation of the barrier certificates for POMDPs applied to privacy verification and robust decision making were disseminated in~\cite{AWLT18,ACJT18}.

The rest of the paper is organized as follows. In the next section, we briefly describe the POMDP formalism. In Section III, we present the problems studied in this paper. In Section~IV, we formulate a hybrid system representation for belief MDP. In Section~V, we show how we can use Lyapunov functions to approximate the reachable belief space. In Section~VI, we propose barrier certificate theorem for safety/optimality verification of POMDPs. In Section VII, we propose a set of sum-of-squares programs for computational implementation of the proposed results. In Section VIII, we elucidate the proposed methodology with examples. Finally, in Section IX, we conclude the paper and give directions for future research.

\textbf{Notation:}
%{\color{blue}
The notations employed in this paper are relatively straightforward. $\mathbb{R}_{\ge 0}$ denotes the set $[0,\infty)$.  $\mathbb{Z}$ denotes the set of integers and $\mathbb{Z}_{\ge c}$ for $c \in \mathbb{Z}$ implies the set $\{c,c+1,c+2,\ldots\}$. $\mathcal{R}[x]$ accounts for the set of polynomial functions with real coefficients in $x \in \mathbb{R}^n$, $p: \mathbb{R}^n \to \mathbb{R}$ and $\Sigma \subset\mathcal{R}$ is the subset of polynomials with a sum-of-squares decomposition; i.e, $p \in  \Sigma[x]$ if and only if there are $p_i \in \mathcal{R}[x],~i \in \{1, \ldots ,k\}$ such that $p = p_i^2 + \cdots +p_k^2$.  For a finite set $A$, $|A|$ denotes the number of elements in $A$ and $co\{A\}$ denotes the convex hull of $A$. In particular, for a family of indexed functions $f_a: \mathcal{X} \to \mathcal{Y}$, $a \in A$, $co\{f_a\}_{a\in A}$ denotes the convex combination of $f_a,~a \in A$, i.e., $co\{f_a\}_{a\in A} = \sum_{a\in A} \beta_a f_a$ where $0 \le \beta_a \le 1$ and $\sum_{a\in A} \beta_a=1$.

%%%%%%%%%%%%%%%%%%%%%%%%%%%%%%%%%%%%%%%%%%%%%%%%%%%%%%%%%%%%%%%%%%%%%%%%%%%%%%%%
\section{Partially Observable Markov Decision Processes}

%In this section, we discuss some preliminary mathematical notions and results that will be employed in the sequel. 

%%%%%%%%%%%%%%%%%%%%%%%%%%%%%%%%%%%%%%%%%%%%%%%%%%%%%%%%%%%%%%%%%%%%%%%%%%%%%%%%

%We denote by $\mathcal{C}^m(X)$, with $X \subseteq \mathbb{R}^n$, the space of $m$-times continuously differentiable functions and by  $\partial^m = \frac{\partial^m}{\partial x^m}$ the derivatives up to order $m$. 

%\section{Preliminaries}

%\subsection{Partially Observable Markov Decision Processes}

% An MDP~\cite{Put94} is a sequential  decision-making modeling framework, in which the actions have stochastic outcomes. 

% \begin{defi}
% An MDP $\mathcal{M}$ is a tuple $(Q,p_0,A,T)$, where
% 	\begin{itemize}
% 		\item $Q$ is a finite set of states with indices $\{1,2,\ldots,n\}$;
% 		\item $p_0:Q\rightarrow[0,1]$ defines the distribution of the initial states, i.e., $p_0(q)$ denotes the probability of starting at $q\in Q$;
% 		\item $A$ is a finite set of actions;
% 		\item $T:Q\times A\times Q\rightarrow [0,1]$ is the transition probability, where 
% 		\begin{multline}
% 		T(q,a,q'):=P(q_t=q'|q_{t-1}=q,a_{t-1}=a),~\\
% 		\forall t\in\mathbb{Z}_{\ge 1}, q,q'\in Q, a\in A. \nonumber
% 		\end{multline}		 
% 	\end{itemize}
% 	\end{defi}

A POMDP is a sequential  decision-making modeling framework, which considers not only the stochastic outcomes of actions, but also the imperfect state observations~\cite{Sondik78}. 

\begin{defi}
A POMDP $\mathcal{P}$ is a tuple $(Q,p_0,A,T,Z,O)$, where

	\begin{itemize}
				\item $Q$ is a finite set of states with indices $\{1,2,\ldots,n\}$;
		\item $p_0:Q\rightarrow[0,1]$ defines the distribution of the initial states, i.e., $p_0(q)$ denotes the probability of starting at $q\in Q$;
		\item $A$ is a finite set of actions;
		\item $T:Q\times A\times Q\rightarrow [0,1]$ is the transition probability, where 
		\begin{multline}
		T(q,a,q'):=P(q_t=q'|q_{t-1}=q,a_{t-1}=a),~\\
		\forall t\in\mathbb{Z}_{\ge 1}, q,q'\in Q, a\in A. \nonumber
		\end{multline}		 
		\item $Z$ is the set of all possible observations. Often, $z\in Z$ is an incomplete projection of the world state $q$, contaminated by  noise;
		\item $O:Q\times A \times Z\rightarrow [0,1]$ is the observation probability (sensor model), where
		\begin{multline}
		O(q,a,z):=P(z_t=z|q_{t}=q,a_{t-1}=a),~\\
	    \forall t\in\mathbb{Z}_{\ge 1}, q\in Q, a\in A, z\in Z. \nonumber
		\end{multline}			 
	\end{itemize}
\end{defi}
	
	Since the states are not directly accessible in a POMDP, decision making requires the history of observations. Therefore, we need to define the notion of a \emph{belief} (or the posterior) as sufficient statistics for the history~\cite{astrom}. Given a POMDP, the belief at $t=0$ is defined as $b_0(q)=p_0(q)$ and $b_t(q)$ denotes the probability of the system being in state $q$ at time $t$. At time $t+1$, when action $a\in A$ is taken and $z \in Z$ is observed, the belief update can be obtained by a Bayesian filter as
%\begingroup\makeatletter\def\f@size{9}\check@mathfonts
\begin{align} \label{equation:belief update}
b_t(q')
&=P(q'|z_t,a_{t-1},b_{t-1}) \nonumber \\
&= \frac{P(z_t|q',a_{t-1},b_{t-1})P(q'|a_{t-1},b_{t-1})}{P(z_t|a_{t-1},b_{t-1})}\nonumber \\
&= \frac{P(z_t|q',a_{t-1},b_{t-1})}{P(z_t|a_{t-1},b_{t-1})} \nonumber \\
&    \times \sum_{q\in Q}P(q'|a_{t-1},b_{t-1},q)P(q|a_{t-1},b_{t-1}) \nonumber \\
&=\frac{O(q',a_{t-1},z_{t})\sum_{q\in Q}T(q,a_{t-1},q')b_{t-1}(q)}{\sum_{q'\in Q}O(q',a_{t-1},z_{t})\sum_{q\in Q}T(q,a_{t-1},q')b_{t-1}(q)},
\end{align}
where the beliefs belong to the belief unit simplex
$$
\mathcal{B} = \left\{ b \in [0,1]^{|Q|} \mid \sum_q b_t(q)=1, \forall t  \right\}.
$$
%\endgroup

A policy in a POMDP setting is then a mapping $\pi:\mathcal{B} \to A$, i.e., a mapping from the continuous belief space into the discrete and finite action space.
 Essentially, a policy defined over the belief space has infinite memory when mapped to the potentially infinite executions of the POMDP.
 
\section{Problem Formulation}\label{sec:probform}

%We represent the uncertainty in the autonomous agent's dynamics as a POMDP with uncertain transition and/or observation probabilities. The class of uncertainties we study belong to an interval~\cite{ITOH2007453}. Let $T_u$ denote the set of triplets $(q,a,q')$  corresponding to the uncertain transition probabilities. Similarly, let $O_u$ denote the set of triplets $(q,a,z)$ corresponding to the uncertain observation probabilities. We consider the class of POMDPs with the following interval transition and/or observation probabilities 
%\begin{subequations}\label{eq:uncertaintpdf}
%		\begin{equation} 
%		T(q,a,q') \in [\underline{l}_{q,a,q'},\overline{l}_{q,a,q'}],~ (q,a,q') \in T_u,
%		\end{equation}	
%		\begin{equation} 
%		O(q,a,z) \in [\underline{o}_{q,a,z},\overline{o}_{q,a,z}],~(q,a,z) \in O_u,
%		\end{equation}	
%		\end{subequations}
%where the constants $0\le\underline{l}_{q,a,q'}\le \overline{l}_{q,a,q'}\le1$ for all $(q,a,q') \in T_u$ and $0\le\underline{o}_{q,a,z}\le \overline{o}_{q,a,z}\le1$ for all $(q,a,z) \in O_u$.
%
%In the sequel, we focus on the case of uncertain transition probabilities, but the extension to the case of uncertain observation transition  probabilities is straightforward and follows the same lines.

%{\color{red} WHY ONLY UNCERTAIN TRANSITION PROBABILITIES?! WHY NOT UNCERTAIN $O$ (SENSOR MODEL)?}

%\subsection{Stochastically Uncertain POMDPs}

%NEEDS MORE WORK, SEE THE MARKOV FIELD MODEL IN~\cite{796wwe3643}!

We consider reachability, safety, and optimality properties of finite POMDPs. In the following, we formally describe the latter problems.
  
  Given an initial belief $b_0$ and a POMDP $\mathcal{P}$, we are concerned with estimating  the set of all belief evolutions emanating from $b_0$, which is known as the  \emph{reachable belief space}.
\begin{problem}[Reachable Belief Space]
Given a POMDP $\mathcal{P}=(Q,b_0,A,T,Z,O)$ as in Definition 1, estimate 
\begin{equation}
\mathcal{R}(b_0)= \left\{ b \in \mathcal{B} \mid b_{t},~\text{for all $t>0$, $a \in A$, and $z \in Z$}\right\}.
\end{equation}
\end{problem}

We are also interested in studying the reachable belief space induced by a policy $\pi$.

\begin{problem}[Reachable Belief Space Induced by Policy $\pi$]
Given a POMDP $\mathcal{P}=(Q,b_0,A,T,Z,O)$ as in Definition~1  and a policy $\pi:\mathcal{B} \to A$, estimate 
\begin{equation}
\mathcal{R}^{\pi}(b_0)= \left\{ b \in \mathcal{B} \mid b_{t},~\text{for all $t>0$ and $z \in Z$}\right\}.
\end{equation}
%\nils{Why is there any estimation here? With a given policy, we just apply to the POMDP and get a finite Markov chain.}
\end{problem}

We define \emph{safety} as the probability of reaching a set of unsafe states $Q_u \subset Q$ being less than a given constant. 
To this end, we use the belief states. We are interested in solving the following problem.
\begin{problem}[Safety at $\tau $]
Given a POMDP as in Definition~1, a point  in time $\tau \in \mathbb{Z}_{\ge0}$, a set of unsafe states $Q_u \subseteq Q$, and a safety requirement constant $\lambda$, check whether
\begin{equation}\label{eq:safety}
\sum_{q \in Q_u}b_{\tau}(q) \le \lambda,
\end{equation}
%where $g:\mathcal{B} \to \mathbb{R}$. In particular, $g$ can be an affine function.
\end{problem}
\smallskip
%Accordingly, we define the unsafe set  as follows
%\begin{equation}\label{eq:unsafeset}
%\mathcal{B}_u = \{ b \in \mathcal{B} \mid b_{\tau}(q) > \lambda,~~q \in Q_u \}.
%\end{equation}
%
%\begin{example}
%Function $g$ in~\eqref{eq:safety} can be an affine function. For example,  $b_t(q_u)$ corresponds to the probability of being in the state $q_u$. Alternatively, $g$ can be a nonlinear function. For instance, $g(b(q_1),b(q_2))=b_t(q_1)b_t(q_2)$ corresponding to the conjunctive probability of being in two states $q_1$ and $q_2$.
%\end{example}

We also consider safety for all time. 

\begin{problem}[Safety]
Given a POMDP as in Definition~1, a set of unsafe states $Q_u \subseteq Q$, and a safety requirement constant $\lambda$, check whether
\begin{equation}\label{eq:safetyalltime}
\sum_{q \in Q_u}b_{t}(q) \le \lambda,~~\forall t \in \mathbb{Z}_{\ge0}.
\end{equation}
%where $g:\mathcal{B} \to \mathbb{R}$. In particular, $g$ can be an affine function.
\end{problem}

In addition to safety, we are interested in checking whether an \emph{optimality} criterion is satisfied.  
\begin{problem}[Optimality]
Given a POMDP as in Definition~1, the reward function $R:Q \times A \to \mathbb{R}$, in which $R(q,a)$ denotes the reward of taking action $a$ while being at state $q$, a point in time $\tau \in \mathbb{Z}_{\ge 0}$, and an optimality requirement $\gamma$, check whether
\begin{equation}\label{eq:optimality}
\sum_{s=0}^{\tau} r(b_s,a_s) \le \gamma,
\end{equation}
where $r(b_s,a_s) = \sum_{q \in Q} b_t R(q,a_t)$.
\end{problem}
\smallskip
%

%Note that both the safety and optimality requirements are affine constraints on the beliefs. 
%We will use this property later in Section~\ref{sec:SOS} to formulate computational tools.

%\section{Main Results}\label{sec:main}

%Checking whether~\eqref{eq:safety} and~\eqref{eq:optimality} hold by solving the POMDP directly is a PSPACE-hard problem~\cite{ChatterjeeCT16}.  In this section, we first demonstrate that POMDPs can be represented as a special hybrid system, i.e., a discrete-time switched system. Then, we borrow a notion from control theory to check the safety and/or optimality requirements of a given POMDP with a guarantee or a certificate. 

\section{Treating POMDPs as Hybrid Systems}

{ Given a POMDP $\mathcal{P}$ as in Definition 1, checking whether~\eqref{eq:safety} and~\eqref{eq:optimality} hold by solving the POMDP directly is undecidable problem~\cite{ChatterjeeCT16}.  In this section, we demonstrate that POMDPs can be represented as a special hybrid system~\cite{goebel2009hybrid}, i.e., a discrete-time switched system~\cite{ahmadi2008non,KUNDU2017191,zhang2009exponential}. We take advantage of this representation in the sequel to formulate indirect methods for analyzing a POMDP.}

According to the belief update equation~\eqref{equation:belief update}, the current beliefs evolve according to the belief values at the previous step, actions at the previous step, and the current observations. Furthermore, the belief evolution \textit{jumps} into different modes based on the actions taken. Therefore, we can represent the belief evolutions as a discrete-time switched system, where the actions $a \in A$ define the switching modes.  Formally, the belief \emph{dynamics}~\eqref{equation:belief update} can be described as
\begin{equation}\label{equation:discretesystem1}
b_t = f_a\left(b_{t-1},z_t\right),
\end{equation}
where $b_t$ denote the belief vector at time $t$ belonging to the belief unit simplex $\mathcal{B}$ and $b_0 \in \mathcal{B}_0 \subset \mathcal{B}$ (in particular, $\mathcal{B}_0$ can be a singleton) representing the set of initial beliefs (prior). In~\eqref{equation:discretesystem1}, $a \in A$ denote the actions that can be interpreted as the switching modes, $z \in Z$ are the observations representing inputs, and $t \in \mathbb{Z}_{\ge 1}$ denote the discrete time instances. The (rational) vector fields $\{f_{a}\}_{a \in A}$ with $f_a: [0,1]^{|Q|} \to [0,1]^{|Q|} $ are described as the vectors with rows
$$
f_a^{q'}(b,\cdot,z) = \frac{O(q',a,z)\sum_{q\in Q}T(q,a,q')b_{t-1}(q)}{\sum_{q'\in Q}O(q',a,z)\sum_{q\in Q},T(q,a,q')b_{t-1}(q)},
$$
where $f_a^{q'}$ denotes the $q'$th row of $f_a$.  

%If the transition probabilities are uncertain, i.e., they belong to some given set, the system can be represented as an uncertain discrete-time switched system
%\begin{equation}\label{equation:discretesystem}
%b_t = f_a\left(b_{t-1},\theta,z_t\right),
%\end{equation}
%where $\theta \in \Theta$ is a set of uncertain parameters and $\Theta$ represents the uncertain transition probability intervals \eqref{eq:uncertaintpdf}. That is, 
%\begin{equation*}
%		\theta_{q,a,q'} = T(q,a,q') \in [\underline{l}_{q,a,q'},\overline{l}_{q,a,q'}],~ (q,a,q') \in T_u,
%\end{equation*}
%and
%\begin{equation}\label{eq:uncertainparams}
%\Theta =\left \{ \theta \mid \theta_{q,a,q'}  \in  [\underline{l}_{q,a,q'},\overline{l}_{q,a,q'}],~ (q,a,q') \in T_u \right \}.
%\end{equation}

A policy then induces regions (partitions) in the belief space where an action is applied. We denote the index of these partitions of the belief space as $\alpha \in \Gamma$, where $\Gamma$ is a finite set.  Then, a policy for a POMDP can be characterized as
\begin{equation}\label{eq:policy}
\pi(b) = 
\begin{cases}
a_1, & b \in \mathcal{B}_{\alpha_1} \\
\vdots & \vdots \\
a_n, & b \in \mathcal{B}_{\alpha_m},
\end{cases}
\end{equation}
inducing the dynamics
\begin{equation}\label{eq:dynpolicy}
b_t = 
\begin{cases}
f_{a_1}(b,z), & b \in \mathcal{B}_{\alpha_1} \\
\vdots & \vdots \\
f_{a_n}(b,z), & b \in \mathcal{B}_{\alpha_m},
\end{cases}
\end{equation}
where $n$ and $m$ are integers. Note that the number of partitions $|\Gamma|=m$ is not necessarily less than or greater than $|A|$. Indeed, the number of $|\Gamma|=m$ depends on the method used for synthesizing the policy. For example, in the value iteration method for POMDPs, $|\Gamma|=m$ is a function of the so called $\alpha$-vectors~\cite{hauskrecht2000value}.

Given the above switched system representation, we can consider two classes of problems in POMDPs:
\begin{itemize}
\item [1.] \textit{Arbitrary-Policy Verification}: This case corresponds to analyzing~\eqref{equation:discretesystem1} under \emph{arbitrary switching} with switching modes given by $a \in A$.
\item [2.] \textit{Fixed-Policy Verification}: This corresponds to analyzing~\eqref{equation:discretesystem1} under \emph{state-dependent switching}. Indeed, the policy $\pi:\mathcal{B} \to A$ determines regions in the belief space where each mode (action) is active.
\end{itemize}

Both cases of switched systems with~\emph{arbitrary switching} and~\emph{state-dependent switching} are well-known in the systems and controls literature~(see~\cite{liberzon2003switching,hespanha2004uniform} and references therein).
{ Note that the second problem is well-known in the probabilistic planning~\cite{DBLP:journals/jair/SteinmetzHB16,kolobov2012planning} and model checking~\cite{Kat16} communities. 
For a fixed policy and a finite state space, the problem reduces to solving a linear equation system to compute the set of reachable states in the underlying Markov chain. }
The next example illustrates the proposed switched system representation for POMDPs with a fixed policy. %{\color{blue} (WE SHOULD CITE TEAM ANDY TEEL AT SOME POINT!)}

\begin{example}Consider a POMDP with two states $\{q_1,q_2\}$, two actions $\{a_1,a_2\}$, and $z = \{z_1, z_2 \} \in {Z}$. Action $a_1$ correspond to remaining in the same state, while action $a_2$ to making a transition to the other state. Observation $z_i$ correspond to the current state being $q_i$. The policy
\begin{equation} \label{eq:example-sds}
\pi=
\left\{\begin{array}{lr}
        a_1, & b \in \mathcal{B}_1,\\
        a_2, & b \in \mathcal{B}_2\\
        \end{array}\right.
\end{equation}
leads to different switching modes based on whether the states belong to the regions $\mathcal{B}_1=\{b \mid 0 \le b(q_1) \le \varepsilon\}$ or $\mathcal{B}_2= \{b \mid  \varepsilon < b(q_1) \le 1\}$ (see Figure~\ref{figure1}). That is, the belief update equation~\eqref{equation:discretesystem1} is given by 
\begin{equation} \label{eq:example-dynamics}
b_t=
\left\{\begin{array}{lr}
        f_{a_1}\left(b_{t-1},z_t\right), & b \in \mathcal{B}_1,\\
        f_{a_2}\left(b_{t-1},z_t\right), & b \in \mathcal{B}_2.\\
        \end{array}\right.
\end{equation}
Note that the belief space is given by $\mathcal{B}=\mathcal{B}_1 \cup \mathcal{B}_2 = \{ b \mid b(q_1)+b(q_2)=1\}$.
\end{example}
\begin{figure}[tbp] 
\begin{center} 
\includegraphics[width=5cm]{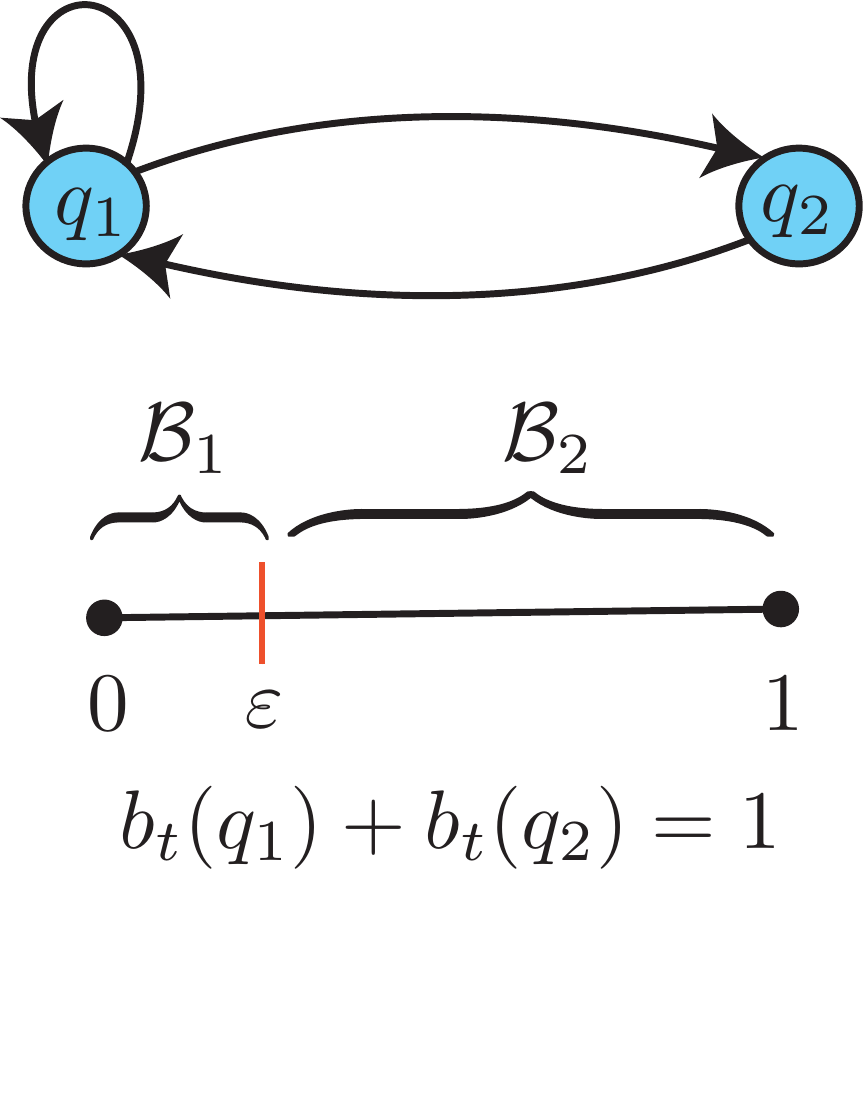}
\vspace{-1cm}
\caption{An example of a POMDP with two states and the state-dependent switching modes induced by the policy~\eqref{eq:example-sds}. }
\label{figure1}
\end{center}
\end{figure}

\section{POMDP Reachability Analysis \\Using Lyapunov Functions}\label{sec:lyapunov}

Based on the switched system representation discussed in the previous section, we can use tools from control theory to analyze the evolution of the solutions of the belief update equation including the reachable belief space without the need to solve the POMDP explicitly. 

From a control theory standpoint, estimating reachable sets, invariant sets, or regions-of-attractions (ROAs) are well-studied problems in the  literature~\cite{blanchini2008set}. Recent results in this area include learning ROAs using Gaussian processes and active learning~\cite{berkenkamp2016safe}, computation of controlled invariant sets for switching affine systems~\cite{legat2018computing}, ROA estimation using formal language~\cite{burchardt2007estimating}, and  finding invariant sets based on barrier functions~\cite{Gurriet2018}. In the following, we propose a basic technique for over-approximating the reachable belief spaces of POMDPs using Lyapunov functions.

%\subsection{Approximating Reachable Belief Space using Lyapunov Functions}\label{sec:lyapunov}

The next theorem uses sub-level sets of Lyapunov functions to  over-approximate the reachable belief space as illustrated in Figure~\ref{figure:innerapprox}.

\begin{figure}[tbp] 
\begin{center} 
\includegraphics[width=6cm]{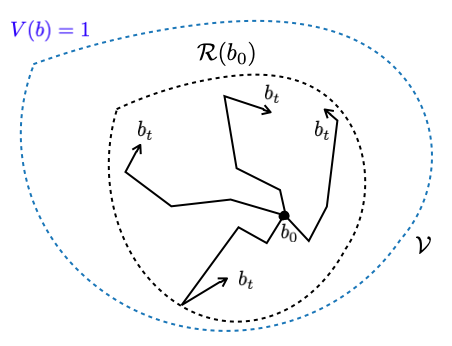}

\caption{Over-approximation of the reachable belief space using sub-level sets of a Lyapunov function. }
\label{figure:innerapprox}
\end{center}
\end{figure}

\begin{thm}\label{theorem:main}
Given the POMDP $\mathcal{P}=(Q,b_0,A,T,Z,O)$ as in Definition~1, consider the belief update dynamics~\eqref{equation:discretesystem1}. If there exists a function $V:\mathcal{B} \to \mathbb{R}$ such that
\begin{multline}\label{eq:con1}
V\left(  f_a\left(b_{t-1},z\right)  \right) - V\left( b_{t-1}  \right)<0,\\~\text{for all}~a \in A,~z\in Z,~b_{t-1} \in \mathcal{V},
\end{multline}
and 
\begin{equation}\label{eq:con2}
b_0 \in \mathcal{V},
\end{equation}
where $\mathcal{V} = \{b \in \mathcal{B} \mid V(b)\le1\}$, then $b_t \in \mathcal{V}$ for all $t \in \mathbb{Z}_{\ge 0}$, and, in particular, $   \mathcal{R}(b_0) \subseteq \mathcal{V}$.
\end{thm}
\begin{proof}
The proof relies on the fact that the belief trajectories point inward at the boundary of $\mathcal{V}$; hence, they remain inside $\mathcal{V}$. We prove by contradiction. Assume, at time $T \neq 0$, we have $b_T \notin \mathcal{V}$ for some $a \in A$ or $z\in Z$. Thus, $V(b_T)>1$. Inequality~\eqref{eq:con1} implies 
\begin{equation*}
V\left(  b_t  \right) - V\left( b_{t-1}  \right)<0,~\text{for all}~a \in A,~z\in Z,~b \in \mathcal{V}.
\end{equation*}
Summing up the above expression from time $1$ to $T$ yields 
\begin{equation*}
\sum_{t=1}^T \left(V\left(  b_t  \right) - V\left( b_{t-1}  \right) \right)<0,~\text{for all}~a \in A,~z\in Z,
\end{equation*}
which simplifies to 
\begin{equation*}
V\left(  b_T  \right) - V\left( b_{0}  \right) <0,~\text{for all}~a \in A,~z\in Z,~b \in \mathcal{V}.
\end{equation*}
That is, $V\left(  b_T  \right) < V\left( b_{0}  \right) $. In addition, from~\eqref{eq:con2}, we have $V\left( b_{0}  \right) \le 1$. Then, it follows that $V\left(  b_T  \right)  \le 1$, which contradict the hypothesis $b_T \notin \mathcal{V}$. 

At this point, we prove by induction that $ \mathcal{R}(b_0) \subseteq \mathcal{V}$. Let $\mathcal{R}_t (b_0)$ denote the reachable belief space at time $t$. At $t=0$, from (13), we have $b_0 \in \mathcal{V}$. Since $\mathcal{R}_0(b_0)=\{b_0\}$, then $\mathcal{R}_0(b_0) \subseteq \mathcal{V}$. At time $t=1$, from (12), we have 
$$
V(b_1)-V(b_0) <0, \forall b_0 \in \mathcal{V},
$$
which implies that $V(b_1)<V(b_0) $ for all $ b_0 \in \mathcal{V}=\{ b \in \mathcal{B} \mid V(b_0)\le 1\}$. Hence, $V(b_1)<V(b_0)\le 1 $ and $b_1 \in \mathcal{V}$. In addition, $\mathcal{R}_1(b_0) = \{b_0,b_1\}$ and since $b_0,b_1 \in \mathcal{V}$, then $\mathcal{R}_1(b_0) \subseteq \mathcal{V}$. Then, by induction we have for any $t=\tau$, $\tau>0$, $b_\tau \in \mathcal{V}$ and $\mathcal{R}_\tau(b_0) \subseteq \mathcal{V}$.
\end{proof}

In practice, we may have a large number of actions for a POMDP. Therefore, it may be difficult to find a single Lyapunov function $V$ that satisfies the conditions of Theorem~\ref{theorem:main}. Next, we show that the computation of the Lyapunov function can be decomposed into finding a set of local Lyapunov functions for each action. This method is  illustrated in Figure~\ref{figure:maxLyap}.

\begin{figure}[tbp] 
\begin{center} 
\includegraphics[width=7cm]{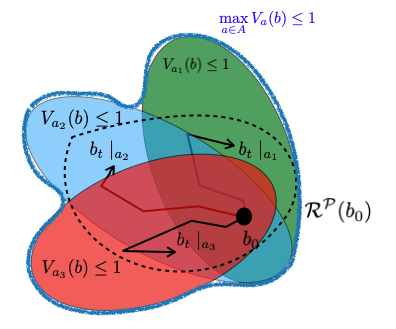}

\caption{Over-approximation of the reachable belief space using the sub-level sets of a set of local Lyapunov functions. }
\label{figure:maxLyap}
\end{center}
\end{figure}

\begin{thm}\label{thm:decomposition}
Given the POMDP $\mathcal{P}=(Q,b_0,A,T,Z,O)$ as in Definition~1, consider the belief update dynamics~\eqref{equation:discretesystem1}. If there exist a set of functions $V_a:\mathcal{B} \to \mathbb{R}$, $a \in A$ such that
\begin{multline}\label{eq:con1x}
V_{a}\left(  f_{a'}\left(b_{t-1},z\right)  \right) - V_{a}\left( b_{t-1}  \right)<0,\\~\forall z\in Z,~b_{t-1} \in \mathcal{V}_{a},
\end{multline}
for all $(a,a') \in A^2$ and 
\begin{equation}\label{eq:con2x}
b_0 \in \mathcal{V}_a,
\end{equation}
for all $a \in A$, where $\mathcal{V}_a = \{b \in \mathcal{B} \mid V_a(b)\le1\}$, then $b_t \in \mathcal{V}_{max}=\{b \in \mathcal{B} \mid \max_{a\in A} V_a(b) \le 1\}$ for all $t \in \mathbb{Z}_{\ge 0}$, and, in particular, $  \mathcal{R}(b_0)  \subseteq \mathcal{V}_{max} $.
\end{thm}
\begin{proof}
We show that if~\eqref{eq:con1x}-\eqref{eq:con2x}  are satisfied then the Lyapunov function $V = \max_{a \in A} V_a$ satisfies the conditions of Theorem~\ref{theorem:main}. If \eqref{eq:con2x} holds for all $a \in A$, we have $V_a(b_0) \le 1$, which in turn implies that $\max_{a\in A} V_a(b_0) \le 1$, as well. Thus, $V = \max_{a \in A} V_a$ satisfies \eqref{eq:con2}.
Furthermore, if~\eqref{eq:con1x} is satisfied for all $(a,a') \in A^2$, we have $V_{a}(f_{a'}(b_{t-1},z)) < V_{a}(b_{t-1})$ for $b_{t-1}\in \mathcal{V}_{a}$, which in turn implies that $\max_{a} V_{a}(f_{a'}(b_{t-1},z)) < \max_{a} V_{a}(b_{t-1})$ for $b_{t-1}$ belonging to $\max_{a} V_{a}(b_{t-1}) \le 1$ and all $a' \in A$. That is, \eqref{eq:con1} is satisfied with $V = \max_{a\in A} V_a$. Then, from Theorem~\ref{theorem:main}, we have $b_t \in \mathcal{V}_{max}$ for all $t \in \mathbb{Z}_{\ge0}$ and $   \mathcal{R}(b_0) \subseteq \mathcal{V}_{max}$.
\end{proof}

If a policy $\pi$ is given, then we can compute a family of local Lyapunov functions. The union of the 1-sublevel sets of  these functions  over-approximates the reachable belief space.

\begin{thm}\label{theorem2}
Given the POMDP $\mathcal{P}=(Q,b_0,A,T,Z,O)$ and the policy $\pi:\mathcal{B} \to A$ and an induced partitioning $\{\mathcal{B}_\alpha\}$ such that $\mathcal{B}=\cup_{\alpha} {\mathcal{B}_\alpha}$ as given in~\eqref{eq:policy}, consider the belief update dynamics~\eqref{equation:discretesystem1}. If there exists a family of functions $V_\alpha:\mathcal{B}_\alpha \to \mathbb{R}$, $\alpha \in \Gamma$, such that
\begin{multline}\label{eq:con11}
V_\alpha\left(  f_a\left(b_{t-1},z\right)  \right) - V_\alpha\left( b_{t-1}  \right)<0,\\~\text{for all}~z\in Z,~b_{t-1} \in \mathcal{V}_\alpha,
\end{multline}
and 
\begin{equation}\label{eq:con22}
b_0 \in \mathcal{V}_\alpha,
\end{equation}
where $\mathcal{V}_\alpha = \{b \in \mathcal{B}_\alpha \mid V_\alpha (b)\le1\}$, then $b_t \in \cup_{\alpha} \mathcal{V}_\alpha$ for all $t \in \mathbb{Z}_{\ge 0}$, and, in particular, $ \mathcal{R}(b_0) \subseteq  \cup_{\alpha \in \Gamma} \mathcal{V}_\alpha $.
\end{thm}
\begin{proof}
From Theorem~\ref{theorem:main}, we can infer that, for each fixed $a \in A$, we have $ \mathcal{R}_a (b_0) \subseteq  \mathcal{V}_\alpha$, where $\mathcal{R}_a (b_0)$ is the reachable belief space under action $a$ starting at $b_0$. Since the support of each $V_{\alpha_i}$ is $\mathcal{B}_{\alpha_i}$, $\mathcal{V}_{\alpha_i} \subseteq \mathcal{B}_{\alpha_i}$, $\cup_{\alpha \in \Sigma}  \mathcal{B}_{\alpha} =\mathcal{B}$, then $\cup_{\alpha \in \Sigma}  \mathcal{V}_{\alpha} =\mathcal{V} \subseteq \mathcal{B}$. Furthermore, since $  \mathcal{R}_a (b_0)  \subseteq  \mathcal{V}_\alpha $, then $  \cup_{a \in A} \mathcal{R}_a (b_0) = \mathcal{R}(b_0) \subseteq  \cup_{\alpha  \in \Gamma}  \mathcal{V}_{\alpha }.$
\end{proof}

We remark at the end of this section that the existence of a policy $\pi$ for a POMDP removes all non-determinism and  partial observability, resulting in
an induced Markov chain. In this case, in~\cite{milias2014optimization}, a method based on Foster-Lyapunov functions was used  for analyzing Markov chains, which may be also adapted to verify POMDPs with a fixed policy.

\section{POMDP Safety Verification \\Using Barrier Certificates}

In the following, we show how we can use barrier certificates to verify properties of the switched systems induced by a POMDP. We begin by verifying safety for a given POMDP.

Let us define the following unsafe set
\begin{equation}\label{eq:unsafeset}
\mathcal{B}_u^s = \left \{ b \in \mathcal{B} \mid \sum_{q \in Q_u} b_{\tau}(q)  > \lambda  \right\},
\end{equation}
which is the complement of~\eqref{eq:safety}.

\begin{thm}\label{theorem-barrier-discrete}
Given the POMDP $\mathcal{P}=(Q,b_0,A,T,Z,O)$ as in Definition~1, consider the belief update dynamics~\eqref{equation:discretesystem1}. Given a set of initial beliefs $\mathcal{B}_0 \subset [0,1]^{|Q|}$, an unsafe set $\mathcal{B}^s_u$ as given in~\eqref{eq:unsafeset} ($\mathcal{B}_0 \cap \mathcal{B}^s_u = \emptyset$), and a constant $\tau \in \mathbb{Z}_{\ge 0}$, if there exists a function $B:\mathbb{Z} \times \mathcal{B} \to \mathbb{R}$ called the barrier certificate such that
\begin{equation}\label{equation:barrier-condition1}
B(\tau,b_{\tau})  > 0, \quad \forall b_{\tau} \in \mathcal{B}^s_u,
\end{equation}
\begin{equation}\label{equation:barrier-condition11}
 B(0,b_0) < 0, \quad \forall b_0 \in \mathcal{B}_0,
\end{equation}
and
\begin{multline}\label{equation:barrier-condition2}
B\left(t,f_a(b_{t-1},z)\right) - B(t-1,b_{t-1}) \le 0, \\ \forall t \in \{1,2,\ldots,\tau\},~\forall a \in A,~\forall z \in Z, ~\forall b \in \mathcal{B},
\end{multline}
then there exist no solution of the belief update equation~\eqref{equation:discretesystem1} such that $b_0 \in \mathcal{B}_0$, and $b_{\tau} \in \mathcal{B}_u^s$ for all $a \in A$.
\end{thm}
\begin{proof}
The proof is carried out by contradiction. Assume at time instance ${\tau}$ there exist a solution to \eqref{equation:discretesystem1} such that $b_0 \in \mathcal{B}_0$ and $b_{\tau} \in \mathcal{B}^s_u$. From inequality~\eqref{equation:barrier-condition2}, we have 
$$
B(t,b_t) \le B(t-1,b_{t-1})
$$
for all $t\in \{1,2,\ldots,{\tau}\}$ and all actions $a \in A$. Hence, $B(t,b_t) \le B(0,b_0)$ for all $t \in \{1,2,\ldots,{\tau}\}$. Furthermore, inequality~\eqref{equation:barrier-condition11} implies that 
$$
B(0,b_0) < 0,
$$
for all $b_0 \in \mathcal{B}_0$.  Since the choice of ${\tau}$ can be arbitrary, this is a contradiction because it implies that $B({\tau},b_{\tau}) \le B(0,b_0) < 0$. Therefore, there exist no solution of \eqref{equation:discretesystem1} such that $b_0 \in \mathcal{B}_0$ and $b_{\tau} \in \mathcal{B}^s_u$ for any sequence of actions $a \in A$. Therefore, the safety requirement is satisfied.
\end{proof}

The above theorem provides conditions under which the POMDP is \emph{guaranteed} to be safe. The next result brings forward a set of conditions, which verifies whether the optimality criterion~\eqref{eq:optimality} is satisfied. This is simply carried out by making the unsafe safe time-dependent (a tube).

\begin{cor}\label{corollary-barrier-discrete:optimality}
Given the POMDP $\mathcal{P}=(Q,b_0,A,T,Z,O)$ as in Definition~1, consider the belief update dynamics~\eqref{equation:discretesystem1} and the optimality criterion $\gamma$ as given by~\eqref{eq:optimality}. Let $\tilde{\gamma}:\mathbb{Z}_{\ge 0} \to \mathbb{R}$ satisfying
\begin{equation} \label{eq:contstraint-on-gamma}
\sum_{s=0}^{\tau} \tilde{\gamma}(s) \le \gamma.
\end{equation}
Given a set of initial beliefs $\mathcal{B}_0 \subset \mathcal{B}$, an unsafe set
\begin{equation}\label{eq:parametrized-unsafe-set}
\mathcal{B}_u^o = \left \{ (t,b) \mid r(b_t,a_t) > \gamma(t) \right\},
\end{equation}
 and a constant $\tau \in \mathbb{Z}_{\ge 0}$, if there exists a function $B:\mathbb{Z} \times \mathcal{B} \to \mathbb{R}$  such that~\eqref{equation:barrier-condition1}-\eqref{equation:barrier-condition2} are satisfied with $\mathcal{B}^o_u$ instead of $\mathcal{B}^s_u$, then for all $b_0 \in \mathcal{B}_0$ the optimality criterion~\eqref{eq:optimality} holds.
\end{cor}
\begin{proof}
The proof is straightforward and an application of Theorem~\ref{theorem-barrier-discrete}. If conditions~\eqref{equation:barrier-condition1}-\eqref{equation:barrier-condition2} are satisfied with $\mathcal{B}^o_u$ instead of $\mathcal{B}^s_u$, based on Theorem~\ref{theorem-barrier-discrete}, we conclude that there exist no solution of the belief update equation~\eqref{equation:discretesystem1} such that $b_0 \in \mathcal{B}_0$, and $b_{\tau} \in \mathcal{B}_u^o$ for all $a \in A$. Therefore, we have
$$
r(b_t,a_t) \le \tilde{\gamma}(t), \quad \forall t \in \{0,1,\ldots,\tau\}.
$$
Summing up both sides of the above equation from $t=0$ to $t=\tau$ yields
$$
\sum_{s=0}^{\tau} r(b_s,a_s) \le \sum_{s=0}^{\tau} \tilde{\gamma}(s).
$$
Then, from \eqref{eq:contstraint-on-gamma}, we conclude that $\sum_{s=0}^{\tau} r(b_s,a_s) \le \gamma$.
\end{proof}

%\begin{figure}[tbp] 
%\begin{center} 
%\includegraphics[width=8.5cm]{Fig2}
%\caption{Two methods for ensuring both safety and optimality. The zero-level sets of $B_s$ ($B_o$) separate the evolutions of the beliefs starting at $\mathcal{B}_0$ from $\mathcal{B}_u^s$  ($\mathcal{B}_u^o$). The red line illustrates the zero-level sets of the barrier certificate formed by taking the maximum of $B_s$ and $B_0$. The blue line illustrate the zero-level set of the barrier certificate formed by taking the convex hull of $B_s$ and $B_0$.}
%\label{figure2}
%\end{center}
%\end{figure}
The technique used in Corollary~\ref{corollary-barrier-discrete:optimality} is analogous to the one used in~\cite{7171125,AHMADI201733} for bounding (time-averaged) functional outputs of systems described by partial differential equations. The method proposed here, however, can be used for a large class of discrete time systems and the belief update equation is a special case that is of our interest.

In practice, we may have a large number of actions. Then, finding a barrier certificate that satisfies the conditions of Theorem~\ref{theorem-barrier-discrete} becomes computationally prohibitive. In the next result, we show how the calculation of the barrier certificate can be decomposed into finding a set of barrier certificates for each action and then taking the convex hull of them.

\begin{thm}\label{theorem-barrier-convexhull}
Given the POMDP $\mathcal{P}=(Q,b_0,A,T,Z,O)$ as in Definition~1, consider the belief update dynamics~\eqref{equation:discretesystem1}. Given a safety constraint as described in \eqref{eq:safety} with a safety requirement $\lambda$, and a point in time $\tau \in \mathbb{Z}_{\ge0}$, if there exists a set of functions $B_a:\mathbb{Z} \times \mathcal{B} \to \mathbb{R}$, $a \in A$, such that
\begin{equation}\label{equation:barrier-condition1x}
B_a(\tau,b_{\tau})  > 0, \quad \forall b_{\tau} \in \mathcal{B}_u^s,~~\forall a\in A,
\end{equation}
with $\mathcal{B}_u^s$ as described in~\eqref{eq:unsafeset},
\begin{equation}\label{equation:barrier-condition11x}
 B_a(0,b_0) < 0, \quad \text{for} \quad b_0=p_0,~~\forall a \in A,
\end{equation}
and
\begin{multline}\label{equation:barrier-condition2x}
B_{a}\left(t,f_{a'}(b_{t-1},z)\right) - B_a(t-1,b_{t-1}) \le 0, \\ \forall t \in \{1,2,\ldots,\tau\},~ \forall (a,a') \in A^2,~\forall z \in Z, ~\forall b \in \mathcal{B},
\end{multline}
then  the safety requirement $\lambda$ is satisfied, i.e., inequality~\eqref{eq:safety}  holds. Furthermore, the overall barrier certificate is given by $B = \text{co}\{B_a\}_{a \in A}$.
\end{thm}
\begin{proof}
We show that if~\eqref{equation:barrier-condition1x}-\eqref{equation:barrier-condition2x} are satisfied then the barrier certificate $B = \text{co}\{B_a\}_{a \in A}$ satisfies the conditions of Theorem~\ref{theorem-barrier-discrete}. For each $a \in A$, we multiply both sides of~\eqref{equation:barrier-condition1x} with a constant $\alpha_a$ such that $\sum_{a\in A} \alpha_a =1$. We obtain
$$
\sum_{a\in A} \alpha_a B_a(\tau,b_{\tau})  > 0, \quad \forall b_{\tau} \in \mathcal{B}_u^s,
$$
which implies that $$B(\tau,b_{\tau}) = \text{co}\{B_a(\tau,b_{\tau})\}_{a \in A}>0,~ \forall b_{\tau} \in \mathcal{B}_u^s.$$ Therefore, \eqref{equation:barrier-condition1x} is satisfied with $B = \text{co}\{B_a\}_{a \in A}$. Similarly, we can show that if~\eqref{equation:barrier-condition11x} is satisfied, $B = \text{co}\{B_a\}_{a \in A}$ satisfies~\eqref{equation:barrier-condition11}. Multiplying both sides of \eqref{equation:barrier-condition2x}  with a constant $\alpha_a$ such that $\sum_{a\in A} \alpha_a =1$ and summing over them yields
\begin{multline*}
\sum_{a\in A} \alpha_a \left( B_a\left(t,f_{a'}(b_{t-1},z)\right) -  B_a(t-1,b_{t-1}) \right) \\
= \sum_{a\in A} \alpha_a  B_a\left(t,f_{a'}(b_{t-1},z)\right) -  \sum_{a\in A} \alpha_a  B_a(t-1,b_{t-1})  \le 0, \\ \forall t \in \{1,2,\ldots,\tau\},~\forall z \in Z,~\forall a' \in A, ~\forall b \in \mathcal{B}.
\end{multline*}
which implies that~\eqref{equation:barrier-condition2} is satisfied for $B = \text{co}\{B_a\}_{a \in A}$. Therefore, from Theorem~\ref{theorem-barrier-discrete}, we conclude that the safety requirement $\lambda$ is satisfied. 
\end{proof}

The efficacy of the above result is that we can search for each action-based barrier certificate $B_a$, $a\in A$, independently or in parallel and then verify whether the overall POMDP satisfies a pre-specified safety requirement (see Fig.~\ref{figure2} for an illustration). 

Next, we demonstrate that, if a policy is given, the search for the barrier certificate can be decomposed into the search for a set of local barrier certificates.  We denote by $a_i$ the action active in the partition $\mathcal{B}_i$. 

\begin{thm}\label{theorem-barrier-policygiven}
Given the POMDP $\mathcal{P}=(Q,b_0,A,T,Z,O)$ as in Definition~1, consider the belief update dynamics~\eqref{equation:discretesystem1}. Given a safety constraint as described in \eqref{eq:safety} with a safety requirement $\lambda$, a point in time $\tau \in \mathbb{Z}_{\ge0}$, and a teaching policy $\pi:\mathcal{B} \to A$ as described in~\eqref{eq:policy}, if there exists a set of function $B_i:\mathbb{Z} \times \mathcal{B}_i \to \mathbb{R}$, $i \in \{1,2,\ldots,N\}$, such that
\begin{equation}\label{equation:barrier-condition1xx}
B_i(\tau,b_{\tau})  > 0, \quad \forall b_{\tau} \in \mathcal{B}_u^o \cap \mathcal{B}_i ,~~i \in \{1,2,\ldots,m\},
\end{equation}
with $\mathcal{B}_u^o$ as described in~\eqref{eq:unsafeset},
\begin{equation}\label{equation:barrier-condition11xx}
 B_i(0,b_0) < 0, \quad \text{for} \quad b_0=p_0,~~i \in \{1,2,\ldots,m\},
\end{equation}
and
\begin{multline}\label{equation:barrier-condition2xx}
B_i\left(t,f_{a_i}(b_{t-1},z)\right) - B_i(t-1,b_{t-1}) \le 0, \\ \forall t \in \{1,2,\ldots,\tau\},~\forall z \in Z, ~\forall b \in \mathcal{B}_i,\\~i \in \{1,2,\ldots,m\},
\end{multline}
then  the safety requirement $\lambda$ is satisfied, i.e., inequality~\eqref{eq:safety}  holds. Furthermore, the overall barrier certificate is given by $B = \text{co}\{B_i\}_{i=1}^m$.
\end{thm}
\begin{proof}
We demonstrate that if~\eqref{equation:barrier-condition1xx}-\eqref{equation:barrier-condition2xx} are satisfied then the barrier certificate $B = \text{co}\{B_i\}_{i =1}^m$ satisfies the conditions of Theorem~\ref{theorem-barrier-discrete}. For each $i \in \{1,2,\ldots,m\}$, we multiply both sides of~\eqref{equation:barrier-condition1xx} with a constant $\alpha_i$ such that $\sum_{i=1}^m \alpha_i =1$. We obtain
$$
\sum_{i=1}^m \alpha_i B_i(\tau,b_{\tau})  > 0, \quad \forall b_{\tau} \in \cup_{i=1}^m \left( \mathcal{B}_u^s \cap \mathcal{B}_i\right).
$$
Since the support of each $B_i$ is $\mathcal{B}_i$, $\cup_{i=1}^m  \mathcal{B}_i =\mathcal{B}$, and $\mathcal{B}_u^s \subset \mathcal{B}$, we have $\cup_{i=1}^m \left( \mathcal{B}_u^s \cap \mathcal{B}_i\right) = \mathcal{B}_u^s \cap \mathcal{B} =  \mathcal{B}_u^s $. Hence, $$B(\tau,b_{\tau}) = \text{co}\{B_i(\tau,b_{\tau})\}_{i =1}^N>0,~ \forall b_{\tau} \in \mathcal{B}_u^s.$$ Therefore, \eqref{equation:barrier-condition1} is satisfied with $B = \text{co}\{B_i\}_{i=1}^m$. Similarly, we can show that if~\eqref{equation:barrier-condition11xx} is satisfied, $B = \text{co}\{B_i\}_{i=1}^m$ satisfies~\eqref{equation:barrier-condition11}. Multiplying both sides of \eqref{equation:barrier-condition2xx}  with constants $\alpha_i$ such that $\sum_{i=1}^m \alpha_i =1$ and summing over them gives
\begin{multline*}
\sum_{i=1}^m \alpha_i \bigg( B_i\left(t,f_{a_i}(b_{t-1},z)\right) -  B_i(t-1,b_{t-1}) \bigg) \\
= \sum_{i=1}^m \alpha_i  B_i\left(t,f_{a_i}(b_{t-1},z)\right) -  \sum_{i=1}^m \alpha_i  B_i(t-1,b_{t-1})  \le 0, \\ \forall t \in \{1,2,\ldots,\tau\},~\forall z \in Z, ~\forall b \in \mathcal{B}.
\end{multline*}
which implies that~\eqref{equation:barrier-condition2} is satisfied for $B = \text{co}\{B_i\}_{i=1}^m$. Therefore, from Theorem~\ref{theorem-barrier-discrete}, we conclude that the safety requirement $\lambda$ is satisfied. 
\end{proof}

\begin{figure}[tbp] 
\begin{center} 
\includegraphics[width=8.5cm]{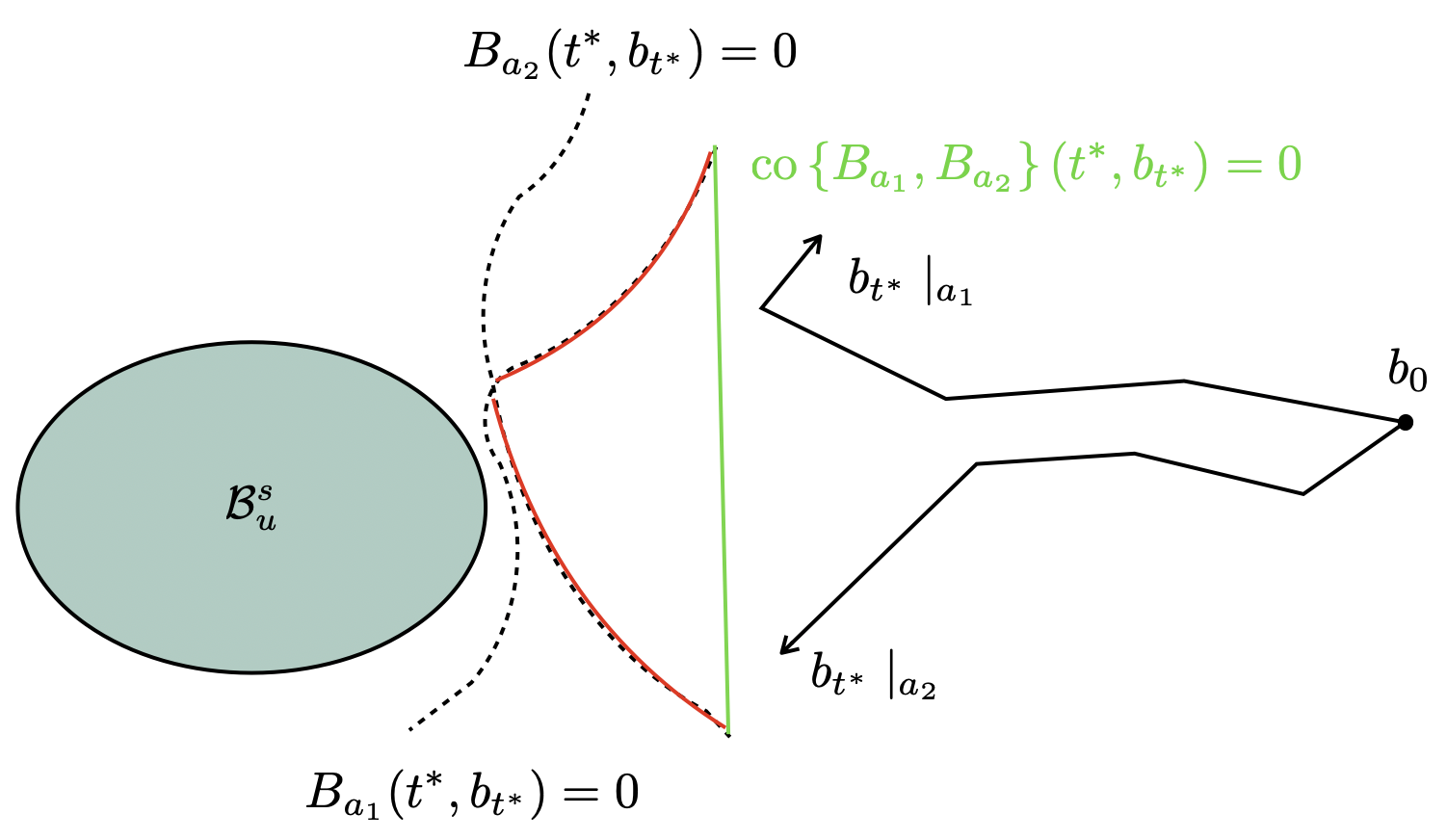}
\caption{Decomposing the barrier certificate computation for a POMDP with two actions $a_1$ and $a_2$: the zero-level sets of $B_{a_1}$ and $B_{a_2}$ at trial $\tau$ separate the evolutions of the hypothesis beliefs starting at $b_0$ from $\mathcal{B}_u^o$. The green line illustrate the zero-level set of the barrier certificate formed by taking the convex hull of $B_{a_1}$ and $B_{a_2}$.}
\label{figure2}
\end{center}
\end{figure}
%The technique used in Corollary~\ref{corollary-barrier-discrete:optimality} is analogous to the one used in~\cite{7171125,AHMADI201733} for bounding (time-averaged) functional outputs of systems described by partial differential equations. The method proposed here, however, can be used for a large class of discrete time systems and the belief update equation is a special case that is of our interest.
%
%In practice, it is often desirable to make sure a design is both optimal and safe. The problem can be described by checking whether the solutions of the belief update switched dynamics~\eqref{equation:discretesystem} enter the following set 
%$$
%\mathcal{B}_u = \mathcal{B}_u^s \cup \mathcal{B}_u^o.
%$$
%To this end, we can adopt either of the following approaches (see Figure~\ref{figure2}). 

We proposed two techniques for decomposing the construction of the barrier certificates and checking  given safety requirements. Our method relied on barrier certificates that take the form of the convex hull of a set of local barrier certificates~(see similar results in~\cite{8264626,2018arXiv180104072A}). Though the convex hull barrier certificate may introduce a level of conservatism, it is computationally easier to find (as will be discussed in more detail in the next section).  

%We remark that another technique that can be used for  decomposition may use non-smooth barrier certificates~\cite{7937882}, i.e., max or min of a set of local barrier certificates. 

%In practice, it is often desirable to make sure a design is both optimal and safe. The problem can be described by checking whether the solutions of the belief update switched dynamics~\eqref{equation:discretesystem} enter the following set 
%$$
%\mathcal{B}_u = \mathcal{B}_u^s \cup \mathcal{B}_u^o.
%$$
%To this end, we can adopt either of the following approaches (see Figure~\ref{figure2}). Both of these approaches are based on the construction of non-smooth barrier certificates. The first one, proposed in~\cite{7937882}, suggests finding the barrier certificate for $\mathcal{B}_u^s$ and $\mathcal{B}_u^o$ separately or in parallel. The barrier certificate for the set $\mathcal{B}_u$ is then the maximum of the two certificates, i.e., $B = \max \{B_s,B_o\}$, where $B_s$ is the barrier certificate for checking safety and $B_o$ is the barrier certificate for checking optimality. The second method proposed by the authors in~\cite{8264626,2018arXiv180104072A} suggests searching for a barrier certificate composed of the convex hull of the $B_s$ and $B_o$. In this paper, we adopt the latter method.

%The following result can be inferred from Theorem~\ref{} in~\cite{}, which states that we can compute the barrier certificates for safety and optimality separately (or in parallel) and then form a non-smooth barrier certificate composed of the two...

\section{Computational Method based on \\ SOS Programming}

In this section, we present a set of SOS programs that can be used to find the Lyapunov functions used for reachability analysis and the barrier certificates used for safety and optimality verification. For a brief introduction to SOS programming and some related results, see Appendix A.

In order to formulate SOS programs, we need the problem data to be polynomial or rational. Fortunately,  the belief update equation~\eqref{equation:belief update} is a rational function in the belief states $b_t(q)$,~$q \in Q$
\begin{multline}\label{eq:belief-update-rational}
b_t(q') = \frac{M_a\left( b_{t-1}(q'),z \right)}{N_a\left( b_{t-1}(q'),z \right)} \\
= \frac{O(q',a_{t-1},z_{t})\sum_{q\in Q}T(q,a_{t-1},q')b_{t-1}(q)}{\sum_{q'\in Q}O(q',a_{t-1},z_{t})\sum_{q\in Q}T(q,a_{t-1},q')b_{t-1}(q)},
\end{multline}
where $M_a$ and $N_a$ are linear and positive functions of $b$.

 Moreover, the safe, and the optimality objectives described by \eqref{eq:safety} and \eqref{eq:optimality}, respectively, and the belief simplex are all  semi-algebraic sets, since they can be described by  polynomial inequalities/equalities. 
 
 Throughout this section, we also assume that the Lyapunov functions and the barrier certificates are  parametrized as polynomial functions of $b$. Hence, the level sets of those functions are semi-algebraic sets, as well.

At this stage, we are ready to present conditions based on SOS programs to find a Lyapunov function for POMDPs and thus  over-approximate the reachable belief space. 

\begin{cor}
Given the POMDP $\mathcal{P}=(Q,b_0,A,T,Z,O)$ as in Definition~1, consider the belief update dynamics~\eqref{eq:belief-update-rational}. If there exist a set of functions $V \in P[b]$ of degree $d$ and $p \in \Sigma[b]$, and a constant $c>0$ such that
\begin{multline}\label{eq:con1-sos} 
- N_a^d(b_{t-1},z)\bigg( V\left( \frac{M_a\left( b_{t-1},z \right)}{N_a\left( b_{t-1},z \right)}  \right) - V\left( b_{t-1}\right) \bigg)  \\ - p(b_{t-1}) \left( 1 - V(b_{t-1})     \right) - c \in \Sigma[b_{t-1}],
\end{multline}
for all $a \in A$ and $z \in Z$, and
\begin{equation}\label{eq:con2-sos}
1-V(b_0) \ge 0,
\end{equation}
 then $b_t \in \mathcal{V}= \{b \in \mathcal{B} \mid V(b)\le1\}$ for all $t \in \mathbb{Z}_{\ge 0}$, and, in particular, $\mathcal{R}(b_0) \subseteq  \mathcal{V}$.
\end{cor}
\begin{proof}
We show that if~\eqref{eq:con1-sos} and~\eqref{eq:con2-sos} hold, then \eqref{eq:con1} and \eqref{eq:con2} are, respectively, satisfied. Condition~\eqref{eq:con2-sos} implies that~\eqref{eq:con2} holds.
Furthermore, condition~\eqref{eq:con1} for system \eqref{eq:belief-update-rational} can be re-written~as  
\begin{equation*}
-\left(V\left(\frac{M_a\left( b_{t-1},z \right)}{N_a\left( b_{t-1},z \right)} \right) - V(b_{t-1}) \right)>0,\\ \forall a \in A, \forall z \in Z.
\end{equation*}
Since $\mathcal{V}=\{b \mid 1-V(b)\ge 0\}$ is a semi-algebraic set, we use Propositions~\ref{chesip} and~\ref{spos} in Appendix~A to obtain
\begin{multline*}
-\bigg(V\left(\frac{M_a\left( b_{t-1},z \right)}{N_a\left( b_{t-1},z \right)} \right) - V(b_{t-1}) \bigg) \\ -  p(b_{t-1})\left(1-V(b_{t-1}) \right) -c  \in \Sigma[b_{t-1}], \forall a \in A, \forall z \in Z.
\end{multline*}
for $p \in  \Sigma[b] $ and $c>0$.
Given that $N_a\left( b_{t-1},z \right)$ is a positive polynomial of degree one, we can relax the above inequality into a sum-of-squares condition given by
\begin{multline*}
-N_a^d(b_{t-1},z)\bigg(V\left(\frac{M_a\left( b_{t-1},z \right)}{N_a\left( b_{t-1},z \right)} \right) - V(b_{t-1}) \bigg) \\ -  p(b_{t-1})\left(1-V(b_{t-1}) \right) -c  \in \Sigma[b_{t-1}], \forall a \in A, \forall z \in Z.
\end{multline*}
Hence, if ~\eqref{eq:con1-sos} holds, then~\eqref{eq:con1}  is satisfied, as well. From Theorem~\ref{theorem:main}, we infer that  $b_t \in \mathcal{V}= \{b \in \mathcal{B} \mid V(b)\le1\}$ for all $t \in \mathbb{Z}_{\ge 0}$, and $ \mathcal{R}(b_0) \subseteq \mathcal{V}$.
\end{proof}

The set $\mathcal{V}$ provides an over-approximation of the reachable belief space $\mathcal{R}(b_0)$. Indeed, we can tighten the over-approximation by solving the following optimization problem:
\begin{eqnarray} \label{eq:optprob}
&\min\limits_{V \in P[b],\gamma>0,p\in \Sigma[b],c>0} \gamma  & \nonumber \\
&V(b_0) - \gamma \le 0,& \nonumber \\
&- V\left(  f_a\left(b,z\right)  \right) + V\left( b  \right) - p(b)\left( \gamma -V(b) \right)-c \in \Sigma[b],&
\end{eqnarray}
for all~$a \in A$ and $z\in Z$.
The above optimization is bilinear in the variables $V$, $\gamma$, and $p$. However, If $p$ is fixed, the problem becomes convex. Similarly, when $V$ and $\gamma$ are fixed, the optimization problem is convex in $p$. Optimization problems similar to~\eqref{eq:optprob} have been proposed in the literature for estimating the region-of-attraction of polynomial (continuous) systems~\cite{packard2010help}. 

We can computationally implement Theorem~\ref{thm:decomposition} by solving $|A|$ optimization problems
\begin{eqnarray} \label{eq:optprob-policy}
&\min\limits_{V_a \in P[b],\gamma>0,p_a\in \Sigma[b],c_a>0} \gamma_a  & \nonumber \\
&V_a(b_0) - \gamma \le 0,& \nonumber \\
&- V_a\left(  f_{a'}\left(b,z\right)  \right) + V_a\left( b  \right) - p_a(b)\left( \gamma_a -V_a(b) \right) \in \Sigma[b],&\nonumber \\
\end{eqnarray}
for all $z\in Z$ and $a' \in A$ and then computing the level set via $\max_{a\in A} V_a(b) \le \gamma$.

Furthermore, we can  implement Theorem~\ref{theorem2} by solving $|\Gamma|$ optimization problems
\begin{eqnarray} \label{eq:optprob-policy}
&\min\limits_{V_\alpha \in P[b],\gamma_\alpha>0,R_\alpha\in \Sigma[b]} \gamma_\alpha  & \nonumber \\
&V_\alpha(b_0) - \gamma_\alpha \le 0,&\nonumber \\
&- V_\alpha\left(  f_a\left(b,z\right)  \right) + V_\alpha\left( b  \right) - R_\alpha(b)\left( \gamma_\alpha -V_\alpha(b) \right) \in \Sigma[b],&  \nonumber \\
\end{eqnarray}
for all~$z\in Z$ in parallel and then computing~$\cup_{\alpha \in \Gamma} \{b \mid V_\alpha(b) \le \gamma_\alpha\}$.

%Once $V$ and $\gamma$ are found through solving the SOS programs, we can check if a reach-safe objective is satisfied. To this end, it suffices to check the following feasibility problem for reachability
%\begin{eqnarray}
%&\mathrm{find}_{W_r \in \Sigma[b]}& \nonumber \\
%&\gamma - V(b(q)) - W_r(b(q)) \left( \sum_{q \in Q_r} b(q) - \delta_1 \right) \in \Sigma[b(q)],& \nonumber \\
%\end{eqnarray}
%which can be derived by applying Proposition A.2.
%If the above feasibility problem has a solution, we infer that, given the initial belief $b_0$, there exists at least one policy that can satisfy the reach objectives.

Next, we  present SOS programs for finding the barrier certificates and check safety and optimality. We begin by an SOS formulation for Theorem~\ref{theorem-barrier-discrete}.

\begin{cor}\label{cor:SOS-Safety}
Given the POMDP $\mathcal{P}=(Q,b_0,A,T,Z,O)$ as in Definition~1, consider the belief update dynamics~\eqref{eq:belief-update-rational}. Given a safety constraint as described in \eqref{eq:safety} with a safety requirement $\lambda$, and a point in time $\tau \in \mathbb{Z}_{\ge 0}$, if there  exist polynomial functions $B \in P[t,b]$ of degree $d$ and   $p^f \in {\Sigma}[b]$, and constants $s_1,s_2>0$ such that
\begin{equation}\label{eq:setssos1}
B\left({\tau},b_{\tau}\right) +   p^f(b_{\tau}) \left( \sum_{q \in Q_u} b_{\tau}(q)  - \lambda   \right)- s_1 \in \Sigma \left[b_{\tau}\right],
\end{equation}
\begin{equation}\label{eq:setssos2}
-B\left(0,b_0\right)  - s_2 >0,
\end{equation}
and 
\begin{multline}\label{eq:setssos3}
- {N_a\left( b_{t-1} \right)}^d\bigg(B\left(t,\frac{M_a\left( b_{t-1},z \right)}{N_a\left( b_{t-1},z \right)} \right) - B(t-1,b_{t-1}) \bigg)   \\  \in \Sigma[t,b_{t-1}],  \forall  t \in \{1,2,\ldots,{\tau}\},~z \in Z,~a \in A,
\end{multline}
then there exists no solution of~\eqref{equation:discretesystem1} such that  $b_{\tau} \in \mathcal{B}_u^s$ and, hence, the safety requirement is satisfied.
\end{cor}
\begin{proof}
 SOS conditions~\eqref{eq:setssos1} and~\eqref{eq:setssos2} are a direct consequence of applying Propositions~\ref{chesip} and~\ref{spos} in Appendix~A to verify conditions \eqref{equation:barrier-condition1} and \eqref{equation:barrier-condition11}, respectively. Furthermore, condition~\eqref{equation:barrier-condition2} for system \eqref{eq:belief-update-rational} can be re-written~as  
\begin{multline*}
B\left(t,\frac{M_a\left( b_{t-1},z \right)}{N_a\left( b_{t-1},z \right)} \right) - B(t-1,b_{t-1})>0,\\ \forall a \in A,~\forall z \in Z.
\end{multline*}
%Since $\theta \in \Theta$ is a semi-algebraic set, we use Propositions~\ref{chesip} and~\ref{spos} in Appendix~A to obtain
%\begin{multline*}
%B\left(t,\frac{M_a\left( b_{t-1}, \theta,z \right)}{N_a\left( b_{t-1},\theta,z \right)} \right) - B(t-1,b_{t-1}) \\ - \sum_{(q,a,q') \in T_u} p^\theta_{q,a,q'}(\theta,b_{t-1})(\underline{l}_{q,a,q'} - \theta_{q,a,q'})(\overline{l}_{q,a,q'}-\theta_{q,a,q'}) \\ \in \Sigma[t,b_{t-1},\theta], \forall a \in A,~\forall z \in Z.
%\end{multline*}
%for $p^\theta_{q,a,q'} \in  \Sigma[b,\theta] $, $(q,a,q') \in T_u$.
Given that $N_a\left( b_{t-1}(q),z \right)$ is a positive polynomial of degree one, we can relax the above inequality into a SOS condition given by
\begin{multline}
- {N_a\left( b_{t-1},z  \right)}^d\bigg(B\left(t,\frac{M_a\left( b_{t-1},z  \right)}{N_a\left( b_{t-1},z  \right)} \right) - B\left(t-1,b_{t-1} \right) \bigg)\\
  \in \Sigma[t,b_{t-1}]. \nonumber
\end{multline}
Hence, if ~\eqref{eq:setssos3} holds, then~\eqref{equation:barrier-condition2}  is satisfied as well. From Theorem~\ref{theorem-barrier-discrete}, we infer that there is no $b_t(q)$ at time $\tau$ such that   $\sum_{q \in Q_u} b_{\tau}(q)  > \lambda$. Equivalently, the safety requirement is satisfied at time $\tau$. That is, $\sum_{q \in Q_u} b_{\tau}(q)  \le \lambda$.
\end{proof}

%Checking whether optimality holds can also be cast into SOS programs. To this end, we assume the reward function is a polynomial (or can be approximated by a polynomial\footnote{This assumption is realistic, since the beliefs belong to a bounded set (a unit simplex) and by Stone-Weierstrass theorem any continuous function defined on a bounded domain can be uniformly approximated arbitrary close by a polynomial~\cite{Stone37}.}) in beliefs , i.e., $R \in \mathcal{R}[b]$.
%
%The following Corollary can be derived using similar arguments as the proof of Corollary~\ref{cor:SOS-Safety}.
%
%\begin{cor}
%Consider the POMDP belief update dynamics~\eqref{eq:belief-update-rational}, the set of uncertain probabilities~\eqref{eq:uncertainparams},  the set of initial beliefs \eqref{eq:intitial-belief-semialgebraic}, and a constant $\tau>0$. If there  exist polynomial functions $\tilde{\gamma} \in \mathcal{R}[t]$ characterizing the unsafe set~\eqref{eq:parametrized-unsafe-set}, $B \in \mathcal{R}[t,b]$ with degree $d$,   $p^u_q \in {\Sigma}[b]$, $q \in Q_u$,  $p_i^0 \in {\Sigma}[b]$, $i = 1,2,\ldots, n_0$, $p^\theta_{q,a,q'} \in  \Sigma[b,\theta] $, $(q,a,q') \in T_u$, and constants $s_1,s_2>0$ such that \eqref{eq:contstraint-on-gamma}, and \eqref{eq:setssos1}-\eqref{eq:setssos3} are satisfied, then for all $b_0 \in \mathcal{B}_0$, the optimality criterion~\eqref{eq:optimality} holds.
%\end{cor}

Similarly, we can formulate SOS feasibility conditions for checking the inequalities in Theorem~5.

\begin{cor}\label{cor:SOS-Safety2}
Given the POMDP $\mathcal{P}=(Q,b_0,A,T,Z,O)$ as in Definition~1, consider the belief update dynamics~\eqref{eq:belief-update-rational}. Given a safety constraint as described in \eqref{eq:safety} with a safety requirement $\lambda$, and a point in time $\tau \in \mathbb{Z}_{\ge 0}$, if there  exist polynomial functions $B_a \in P[t,b]$, $a \in A$, of degree $d$ and   $p^f_a \in {\Sigma}[b]$, $a \in A$, and constants $s^1_a,s^2_a>0$, $a \in A$, such that
\begin{multline}\label{eq:setssos12}
B_a\left({\tau},b_{\tau}\right) +   p^f_a(b_{\tau}) \left( \sum_{q \in Q_u} b_{\tau}(q)  - \lambda    \right) \\
- s^1_a \in \Sigma \left[b_{\tau}\right],~~a \in A,
\end{multline}
\begin{equation}\label{eq:setssos22}
-B_a\left(0,b_0\right)  - s_a^2 >0,~~a \in A,
\end{equation}
and 
\begin{multline}\label{eq:setssos32}
- {N_a\left( b_{t-1} \right)}^d\bigg(B_{a}\left(t,\frac{M_{a'}\left( b_{t-1},z \right)}{N_{a'}\left( b_{t-1},z \right)} \right) - B_{a}(t-1,b_{t-1}) \bigg)   \\  \in \Sigma[t,b_{t-1}],  \forall  t \in \{1,2,\ldots,{\tau}\},~z \in Z,~(a,a') \in A^2,
\end{multline}
then there exists no solution of~\eqref{eq:belief-update-rational} such that  $b_{\tau} \in \mathcal{B}_u^s$ and, hence, the safety requirement is satisfied.
\end{cor}

We assume that a policy in the form of~\eqref{eq:policy} assigns actions to semi-algebraic partitions of  the hypothesis belief space $\mathcal{B}$ described as
\begin{equation}
\mathcal{B}_i =\left\{  b \in \mathcal{B} \mid g_i(b) \le 0      \right\},~~i \in \{1,2,\ldots,m\}.
\end{equation}
We then have the following SOS formulation for Theorem~6 using Positivstellensatz (Proposition~1 in Appendix~A).

\begin{cor}\label{cor:SOS-Safety3}
Given the POMDP $\mathcal{P}=(Q,b_0,A,T,Z,O)$ as in Definition~1, consider the belief update dynamics~\eqref{eq:belief-update-rational}. Given a safety constraint as described in \eqref{eq:safety} with a safety requirement $\lambda$, and a point in time $\tau \in \mathbb{Z}_{\ge 0}$, a  policy $\pi:\mathcal{B} \to A$ as described in~\eqref{eq:policy}, if there  exist polynomial functions $B_i \in P[t,b]$, $i \in \{1,2,\ldots,m\}$, of degree $d$, $p^{l_1}_i \in {\Sigma}[b]$, $i \in \{1,2,\ldots,m\}$, $p^{l_2}_i \in {\Sigma}[b]$, $i \in \{1,2,\ldots,m\}$, $p^{l_3}_i \in {\Sigma}[b]$, $i \in \{1,2,\ldots,m\}$, and   $p^f_i \in {\Sigma}[b]$, $i \in \{1,2,\ldots,m\}$, and constants $s^1_i,s^2_i>0$, $i \in \{1,2,\ldots,m\}$, such that
\begin{multline}\label{eq:setssos12}
B_i\left({\tau},b_{\tau}\right) +   p^f_i(b_{\tau}) \left(  \sum_{q \in Q_u} b_{\tau}(q)  - \lambda   \right) +p^{l_1}_i(b_{\tau})g_i(b_{\tau}) \\
- s^1_i \in \Sigma \left[b_{\tau}\right],~~i \in \{1,2,\ldots,m\},
\end{multline}
\begin{equation}\label{eq:setssos22}
-B_i\left(0,b_0\right) +p^{l_2}_i(p_0)g_i(p_0)  - s_i^2 >0,~~i \in \{1,2,\ldots,m\},
\end{equation}
and 
\begin{multline}\label{eq:setssos32}
- {N_a\left( b_{t-1} \right)}^d\bigg(B_{i}\left(t,\frac{M_a\left( b_{t-1},z \right)}{N_a\left( b_{t-1},z \right)} \right) - B_{i}(t-1,b_{t-1}) \bigg)   \\  +p^{l_3}_i(b_{t-1})g_i(b_{t-1})  \in \Sigma[t,b_{t-1}],  \forall  t \in \{1,2,\ldots,{\tau}\},\\~z \in Z,~a \in A,~~i \in \{1,2,\ldots,m\},
\end{multline}
then there exists no solution of~\eqref{eq:belief-update-rational} such that  $b_{\tau} \in \mathcal{B}_u^s$ and, hence, the safety requirement is satisfied.
\end{cor}

Checking whether optimality holds can also be cast into sum-of-squares programs. To this end, we assume the reward function is a polynomial (or can be approximated by a polynomial\footnote{This assumption is realistic, since the beliefs belong to a bounded set (a unit simplex) and by Stone-Weierstrass theorem any continuous function defined on a bounded domain can be uniformly approximated arbitrary close by a polynomial~\cite{Stone37}.}) in beliefs , i.e., $R \in P[b]$.

The following Corollary can be derived using similar arguments as the proof of Corollary~\ref{cor:SOS-Safety}.

\begin{cor}
Given the POMDP $\mathcal{P}=(Q,b_0,A,T,Z,O)$ as in Definition~1, consider the belief update dynamics~\eqref{eq:belief-update-rational}. Given a constant $\tau \in \mathbb{Z}_{\ge 0}$ and the optimality requirement~\eqref{eq:optimality}, if there  exist polynomial functions $\tilde{\gamma} \in P[t]$ characterizing the unsafe set~\eqref{eq:parametrized-unsafe-set}, $B \in P[t,b]$ with degree $d$,   $p^u_q \in {\Sigma}[b]$, $q \in Q_u$,  $p_i^0 \in {\Sigma}[b]$, $i = 1,2,\ldots, n_0$, and constants $s_1,s_2>0$ such that \eqref{eq:contstraint-on-gamma}, and \eqref{eq:setssos1}-\eqref{eq:setssos3} are satisfied, then  the optimality criterion~\eqref{eq:optimality} holds.
\end{cor}

\section{Case Studies}

In this section, we illustrate the proposed method using two examples.{ In the first example, we study a toy example of ad scheduling scenario in social media using a POMDP model. We illustrate how the over-approximation method based on local Lyapunov functions introduced in Section V can be used to predict the outcome of an ad scheduling policy.  In the second example, we compare the teaching performance of two recently proposed machine teaching algorithms (myopic  and Ada-L)~\cite{chen18adaptive} and show the safety verification method presented in Section VI can be used to show that myopic teaching has poor performance while Ada-L has guaranteed convergence.}

\subsection{Interactive Advertisement Scheduling}
In personalized live social media, a website's revenue depends on   user engagement, which is closely related to the user's interest to the streamed online contents \cite{krishnamurthy2018multiple}. User interest evolves probabilistically with the online contents and the ads inserted \cite{yadati2014cavva}. Furthermore, user interest to online content is not directly observable but can only be inferred from the number of views or the number of likes during a given time interval \cite{lehmann2012models}. The broadcaster's objective is to estimate user interest and schedule relevant  ads from time to time to maintain high user interest and engagement.  

Such interactive ads scheduling can be represented as a sequential decision-making problem under partial observability \cite{krishnamurthy2018multiple}, where a POMDP $\mathcal{P}=(Q,p_0,A,T,Z,O)$ can describe how the user's interest evolves. Each state $q\in Q$ represents different levels of interest where $p_0$ denotes the initial  interest distribution. The actions available to the broadcaster are to insert different ads or continue with the live stream. The transition $T(q,a,q')$ describes how a user's interest may evolve  depending on the ads inserted or the content of live stream when no ads is scheduled. The observation $Z$ denotes the number of likes in a given time interval, which can be abstracted into a finite number of categories. The observation function $O(q,z)$ denotes the probability of observing $z$ likes when the user's interest level is at $q$. The transition and observation functions can be obtained from data \cite{krishnamurthy2018multiple}.

We consider a concrete POMDP model $\mathcal{P}$ where the topology of the underlying MDP is as shown in Figure~\ref{fig:ad}. The user interest has three levels --- low, medium, and high. The initial probability distribution of the user interest is $p_0(q_1)=1$, i.e., the user initially has low interest. The actions available are $A=\{a_0,a_1\}$, where $a_0$ denotes no ads and $a_1$ denotes scheduling ads. The transition probabilities $T$ with respect to each action are as shown in (\ref{eq:ad}). 

\begin{align}\label{eq:ad}
T_{a_0}&=\begin{bmatrix}
   0.8, &0.2 ,&0.1 \\
   0.1,&0.7,&0.2 \\
   0.1,&0.1,&0.7
\end{bmatrix},
T_{a_1}=\begin{bmatrix}
0.5, &0.3, &0.2 \\
0.3, &0.6,    &0.2 \\
0.2, &0.1,  &0.6
\end{bmatrix}.  
\end{align}

The observations are number of likes in a give time interval and can be represented by a Poisson process \cite{krishnamurthy2018multiple}. To have a finite number of observations, we have $Z=\{z_1,z_2,z_3\}$ to denote low, medium, and high number of likes with the thresholds $\gamma_1$ and $\gamma_2$. The observation functions are then defined as follows.
\begin{align}\label{eq:obs}\nonumber
O(q_i,z_1)&=\sum_{j\leq\gamma_1} \frac{\lambda_i^j\exp(-\lambda_i)}{j!},\\ 
O(q_i,z_2)&=\sum_{\gamma_1<j\leq\gamma_2} \frac{\lambda_i^j\exp(-\lambda_i)}{j!},\\\nonumber
O(q_i,z_3)&=\sum_{j>\gamma_2}\frac{\lambda_i^j\exp(-\lambda_i)}{j!},
\end{align}
where $\lambda_i$ represents the rate of Poisson process when the interest is in state $q_i$. In this particular example, $\lambda_1=2$, $\lambda_2=4$, $\lambda_3=6$, $\gamma_1=3$ and $\gamma_2=6$. Then from (\ref{eq:obs}), it can be found that
\begin{align}\nonumber
O(q_1,z_1)&=0.8571, O(q_1,z_2)= 0.1383, O(q_1,z_3)= 0.0046,\\ \nonumber
O(q_2,z_1)&= 0.4335, O(q_2,z_2)= 0.4559, O(q_2,z_3)=0.1106,\\\nonumber
O(q_3,z_1)&= 0.1512, O(q_3,z_2)= 0.4551, O(q_3,z_3)=0.3937.
\end{align}

\begin{figure}[t]
	\centering	
	\begin{tikzpicture}[shorten >=1pt,node distance=4cm,on grid,auto, thick,scale=.55, every node/.style={transform shape}]
	\node[state] (q_1)   {$q_1$};
	\node[state] (q_2) [below left = 6cm of q_1] {$q_2$};
	\node[state] (q_3) [below right =6cm of q_1] {$q_3$};
	\node[text width=3cm] at (0,0) {low};
	\node[text width=3cm] at (-3.5,-5) {Medium};
	\node[text width=3cm] at (5.5,-5) {High};
	%\node[state] (s_4) [below=of s_3] {};
	%	\node[state] (s_5) [left=of s_4] {};
	%	\node[state] (s_6) [left=of s_5] {};
	\path[->]
	(q_1) edge [pos=0.5, loop, above=0.1] node  {} (q_1)
	(q_1) edge [pos=0.5, bend right, above=0.5,sloped] node {} (q_2)
	(q_1) edge [pos=0.5, bend left, above=0.5,sloped] node {} (q_3)
	
	(q_2) edge [pos=0.5, loop left, left=0.1] node  {} (q_2)
	(q_2) edge [pos=0.5, bend right, above=0.5,sloped] node {} (q_1)
	(q_2) edge [pos=0.5, above=0.5] node {} (q_3)
	
	(q_3) edge [pos=0.5, loop right, right=0.1] node {} (q_3)
	(q_3) edge [pos=0.5, bend left, above=0.5,sloped] node {} (q_1)
	(q_3) edge [pos=0.5, bend left, above=0.5,sloped] node {} (q_2)
	;
	
	%(s_3) edge node [pos=0.5, sloped, above]{$D_1 open$} (s_4)	
	%	(s_4) edge node [pos=0.5, sloped, above]{$R_2 to 1$} (s_5)
	%	(s_5) edge node [pos=0.5, sloped, above]{$R_2 in 1$} (s_6)
	%	(s_3) edge node [pos=0.5, sloped, above]{$r_1$} (s_0);	
	\end{tikzpicture}
	\caption{The underlying MDP in Example I}
	\label{fig:ad}
\end{figure}
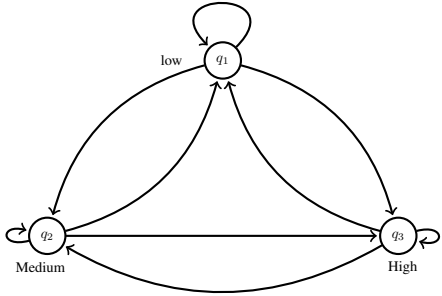

In this example, we study the reachable belief space given the following policy
\begin{equation}
\pi(b) = \begin{cases}
a_0, & b \in \mathcal{B}_1,\\
a_1, & b \in \mathcal{B}_2,
\end{cases}
\end{equation}
where $\mathcal{B}_1 = \{ b \mid b(q_1)+b(q_2) \le 0.5\}$ and $\mathcal{B}_2 = \{ b \mid b(q_1)+b(q_2) > 0.5\}$. That is, if the probability of the user's interest being low or medium is less than $0.5$ show ads, and  if the probability of the user's interest being high is greater than $0.5$ show no ads (see Figure~\ref{fig:partition}).

  In this example, we assume we are not sure about the interest level of the user. Therefore, we assign equal initial probability to all three states ($b_0(q_1)=b_0(q_2)=b_0(q_3)=1/3$). We are then interested in checking whether starting from $b_0$, given the actions and observations, we can reach the state with high interest $q_3$ ($b_t(q_3)=1$). To this end, we calculate the reachable belief spaces using the sub-level sets of Lyapunov functions using optimization problem~(34) and compare with simulations of the POMDP belief state evolutions. We fix the degree $d$ of variables $V_a$, $a=a_0,a_1$ and solve the optimization problem. The polynomial variable $p_a =\sum_{q\in\{q_1,q_2\}} b(q)^d$. Note that $b \in \mathcal{B}$, so $p_a$ is a positive polynomial.

\begin{figure}[tpb]
	\centering
	\begin{tikzpicture}[scale=.9]
	%\draw[gray!50, thin, step=1] (0,0) grid (5,5);
	\draw[very thick,->] (0,0) -- (5.2,0) node[right] {$b(q_1)$};
	\draw[very thick,->] (0,0) -- (0,5.2) node[above] {$b(q_2)$};
	
	\draw (0,0.05) -- (0,-0.05) node[below] {\footnotesize 0};
	\draw (1,0.05) -- (1,-0.05) node[below] {\footnotesize 0.2};
	\draw (2,0.05) -- (2,-0.05) node[below] {\footnotesize 0.4};
	\draw (3,0.05) -- (3,-0.05) node[below] {\footnotesize 0.6};
	\draw (4,0.05) -- (4,-0.05) node[below] {\footnotesize 0.8};
	\draw (5,0.05) -- (5,-0.05) node[below] {\footnotesize1};
	
	\draw (-0.05,1) -- (0.05,1) node[left] {\footnotesize 0.2};
	\draw (-0.05,2) -- (0.05,2) node[left] {\footnotesize 0.4};
	\draw (-0.05,3) -- (0.05,3) node[left] {\footnotesize 0.6};
	\draw (-0.05,4) -- (0.05,4) node[left] {\footnotesize 0.8};
	\draw (-0.05,5) -- (0.05,5) node[left] {\footnotesize 1};
	
	\draw (0,5) -- node[below,sloped] {} (5,0);
	
	\draw[dashed] (0,2.5) -- node[below,sloped] {} (2.5,0);
	
	%\draw (0.1,0.8) -- (0.5,0.8) -- (0.5,1.9) -- (0.1,1.9) -- (0.1,0.8);
	
	%\draw[dashed] (1.6,0.2) -- (3.25,0.2) -- (3.25,0.8) -- (1.6,0.8) -- (1.6,0.2);
	
	%\draw (0,2) -- (1,2) -- (1,3) -- (0,3) -- (0,2);
	
	%\draw (0,2) parabola (0.1,0.8);
	
	%\draw (1,3) parabola (0.5,1.9);
	
	\draw[->] (1.35,1.35) parabola (3,3) ;
	
	%\draw[dashed] (1,3) parabola (3.25,0.8);
	
	\node[] at  (3.5,0.6) {$\mathcal{B}_2$};
	
	\node[] at   (0.6,0.6) {$\mathcal{B}_1$};
	
	\node[] at   (3.5,3.2) {$b(q_1)+b(q_2)=0.5$};
%	\node[] at   (-0.5,1.5) {$q_1$};
%	\node[] at   (-0.5,2.5) {$q_2$};
%	\node[] at   (-0.5,3.5) {$q_6$};
%	\node[] at   (-0.5,4.5) {$q_{10}$};
%	\node[] at   (1.5,-0.5) {$q_3$};
%	\node[] at   (2.5,-0.5) {$q_4$};
%	\node[] at   (3.5,-0.5) {$q_9$};
%	\node[] at   (4.5,-0.5) {$q_{14}$};
%	\node[] at   (1.5,1.5) {$q_5$};
%	\node[] at   (1.5,2.5) {$q_7$};
%	\node[] at   (2.5,1.5) {$q_8$};
%	\node[] at   (1.5,3.5) {$q_{11}$};
%	\node[] at   (2.5,2.5) {$q_{12}$};
%	\node[] at   (3.5,1.5) {$q_{13}$};
	
	\fill[blue!50!cyan,opacity=0.3] (0,5) -- (5,0) -- (2.5,0) -- (0,2.5) -- cycle;
	
	\fill[red!50!cyan,opacity=0.3] (0,0) -- (2.5,0) -- (0,2.5) -- cycle;
%	
%	\fill[red!50!cyan,opacity=0.3] (1,4) -- (2,4) -- (2,2) -- (1,2) -- cycle;
%	
%	\fill[red!50!cyan,opacity=0.3] (2,3) -- (3,3) -- (3,1) -- (2,1) -- cycle;
%	
%	\fill[red!50!cyan,opacity=0.3] (3,2) -- (4,2) -- (4,0) -- (3,0) -- cycle;
%	
%	\fill[red!50!cyan,opacity=0.3] (4,1) -- (5,1) -- (5,0) -- (4,0) -- cycle;
%	\draw[black,fill=black] (1.5,0.5)  circle (.5ex);
%	
%	\draw[black,fill=black] (0,2)  circle (.3ex);
%	\draw[black,fill=black] (1,3)  circle (.3ex);
%	
	%\foreach \x in {0,1,2,3,4,5} \draw (\x,0.05) -- (\x,-0.05) node[below] {\tiny(\x*0.2};
	%\foreach \y in {0,0.2,0.4,0.6,0.8,1} \draw (-0.05,\y) -- (0.05,\y) node[left] {\tiny\y};
	
	%\fill[blue!50!cyan,opacity=0.3] (8/3,1/3) -- (1,2) -- (13/3,11/3) -- cycle;
	
	%\draw (-1,4) -- node[below,sloped] {\tiny$x_1+x_2\geq3$} (5,-2);
	%\draw (1,-3) -- (3,1) -- node[below left,sloped] {\tiny$2x_1-x_2\leq5$} (4.5,4);
	%\draw (-1,1) -- node[above,sloped] {\tiny$-x_1+2x_2\leq3$} (5,4);
	% \draw[step=0.1,black,thin] (0,0) grid (1,1);
	\end{tikzpicture}
	\caption{The belief space in the interactive advertising example. The policy assigns actions $a_0$ (no ads) and $a_1$ (show ads) to different regions of the belief space $\mathcal{B}_1$ (high interest) and $\mathcal{B}_2$ (low to medium interest), respectively.}
	\label{fig:partition}
	%\vspace{-3.5mm}
\end{figure}
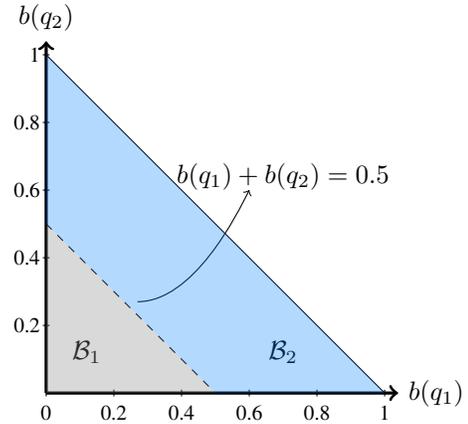

\begin{figure}[t]
	\centering	
\includegraphics[width=6cm]{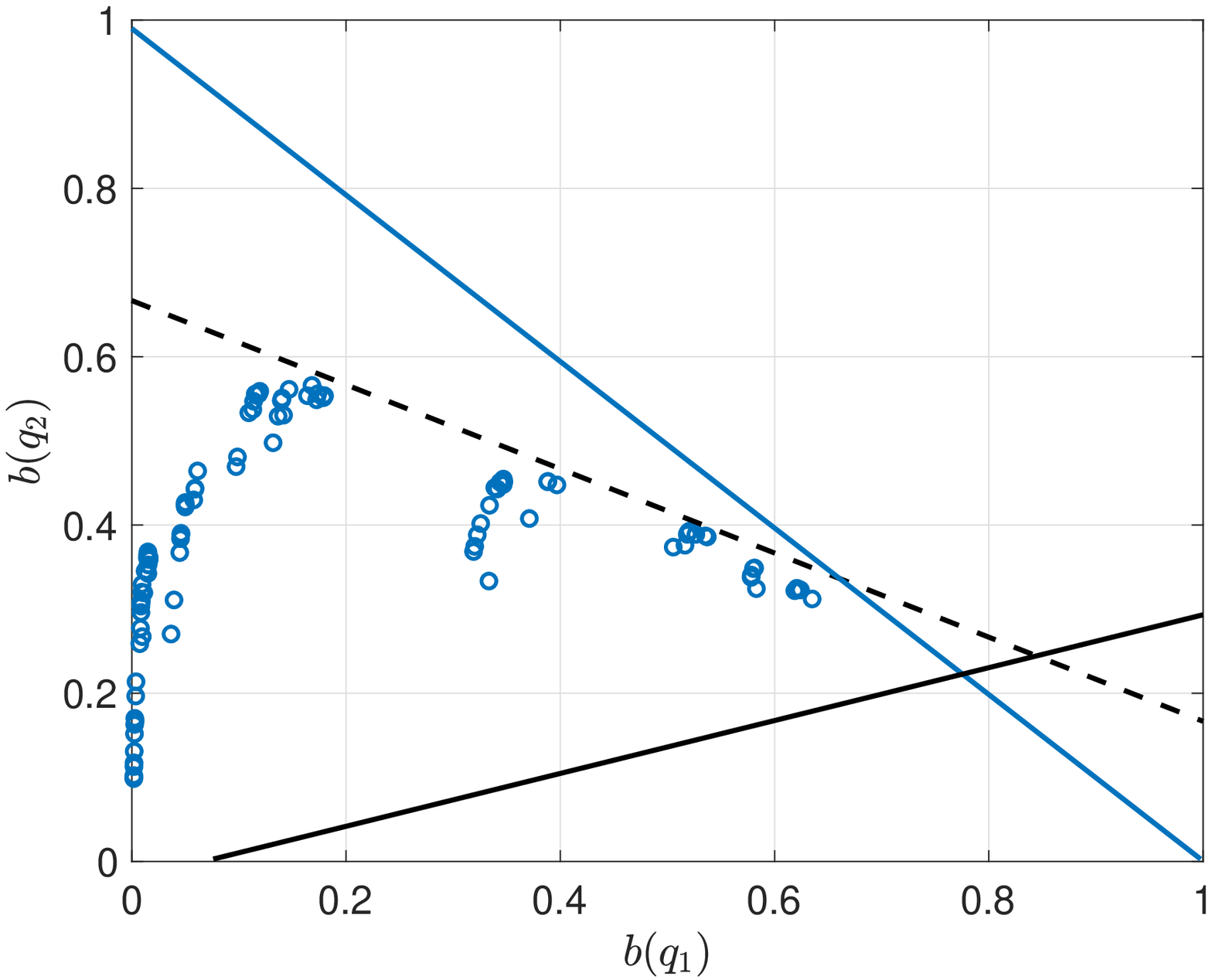}
\includegraphics[width=6cm]{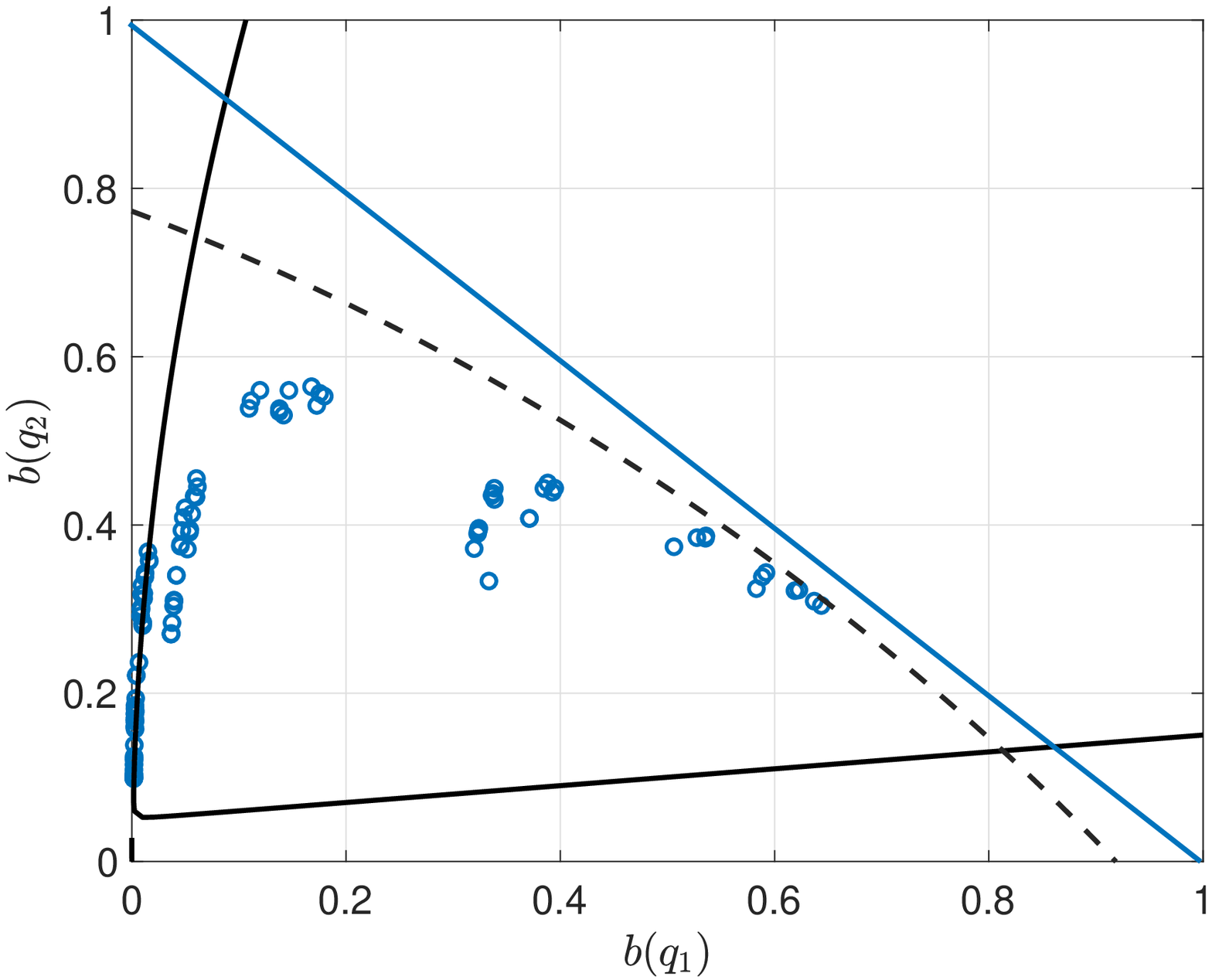}
	\caption{The solutions of the POMDP (blue circles) in Case Study A and the approximated reachable belief space (contained between the black lines) for local Lyapunov functions of degree one (top) and degree three (below). The solid blue line demonstrates the boundary of the belief simplex.}
	\label{fig:reachable}
\end{figure}

Figure~\ref{fig:reachable} illustrates the obtained results for local Lyapunov functions of degree one $d=1$ and degree three $d=3$. As it can be seen in the figure, simulations of the POMDP shows that the state $b(q_3)=1$ (i.e., $b(q_1)=b(q_2=0)$) is not reachable even with the policy $\pi$ due to the uncertainty of the problem and partial observation in the interest levels of the users. The over-approximation of the reachable belief space by local Lyapunov functions of degree one
 fails to capture this property; whereas, the over-approximation using local Lyapunov functions of degree three shows that the state $b(q_1)=b(q_2)=0$ is not reachable for the POMDP. Therefore, one can check reachability properties of POMDPs without the need for solving them explicitly by considering a sufficiently large degree of the local Lyapunov functions.

\subsection{Machine Teaching}

In recent years, the need for machine learning systems has far exceeded the supply of machine learning experts.  Machine teaching, that is, algorithms designed to enable machines to teach other machines or humans, have thus received attention~\cite{simard2017machine}. Formally, machine teaching characterizes the algorithmic framework of designing an optimal training set for learning a target hypothesis~\cite{zhu2015machine}. The target hypothesis is \emph{given} to the algorithm and the goal of the teacher (machine) is to generate a minimal sequence of training examples such that the target hypothesis can be learned by a learner (human or another machine) from a finite set of hypotheses. 

We represent machine teaching as a sequential decision making under uncertainty scenario. To this end, we use a POMDP representation for the learner based on the state-dependent teaching model (for more details see~\cite{AWCYT19}). The POMDP model can be described as follows.

%\yuxin{This setting doesn't fit exactly into the version space model we consider in the reference. The version space model is actually an MDP instead of a POMDP, i.e., we assume that the teacher observes the learner's state $h_{t}$ at each time step $t$. In such cases, $p_0$ is a delta distribution, }
\begin{defi}[Learning POMDP~\cite{AWCYT19}]
  The  learning POMDP $\mathcal{P}_L$ is a tuple $(\mathcal{H}, p_0, \mathcal{X}, T, \mathcal{Y}, O)$ 
  %\yuxin{$\mathcal{Y}$ is fully known to the teacher, because the teacher knows $\hstar$. The way that the current POMDP is described looks like the model of an (active) learning system, rather than a machine teaching model. The current model is fine for the non-adaptive teaching setting. We will have to change the description $b_t$ to be the belief of the teacher (as I commented below in the next paragraph). For the adaptive setting, the teacher's belief on the learner's hypothesis also depends on the learner's feedback---I think one possible fix is to enrich the observation set $\mathcal{Y}$, such that it also captures the feedback that the teacher receives from the learner.}, where
\begin{itemize}
\item the hypotheses set $\mathcal{H}$ is a finite set of hidden states;
\item $p_0$ is the probability of having an initial hypothesis $h_0 \in \mathcal{H}$;
\item the set of examples $\mathcal{X}$ constitute the finite set of actions;
\item the set of labels $\mathcal{Y}$ constitute the finite set of observations;
\item $T$ describes the transitions from one hypothesis (state) to another characterized by a preference function;
\item $O(y_t \mid h_t,x_t)$ determines the probability of seeing a label (observation) given the current hypothesis and example.
\end{itemize}
\end{defi}

The learner starts with an initial hypothesis $h_0$, i.e., $b(h_0)=1$. Over a sequence of trials, in which an example $x_t \in \mathcal{X}$ is shown and the learner receives a corresponding observation $y_t \in \mathcal{Y}$, the learner updates its belief in the hypotheses.

We assume all  hypotheses are uniformly distributed, or equivalently, $O(y _t \mid h_t, x_t)$ is binary. Moreover,  the transition function $T(h, x_{t-1}, h')$ defines a \emph{uniform} distribution: the learner only goes to the hypotheses $h'$ that are the most preferred; hence, $T$ induces a uniform distribution over the most preferred hypothesis according to a preference function $\sigma$.

%\yuxin{The description below fits the active learning setting, but not teaching. I think it is important that the algorithm takes the perspective of the teacher instead of the learner}
The learner starts with an initial hypothesis $h_0$ and over a sequence of trials, in which an example $x_t\in \mathcal{X}$ is shown and the learner receives a corresponding observation $y_t \in \mathcal{Y}$, develops a belief in the new hypothesis $h$.

\begin{figure*}[!t]
  \centering
  \begin{subfigure}[b]{1.0\textwidth}
    \includegraphics[trim={8pt 5pt 8pt 5pt},height=.11\textwidth]{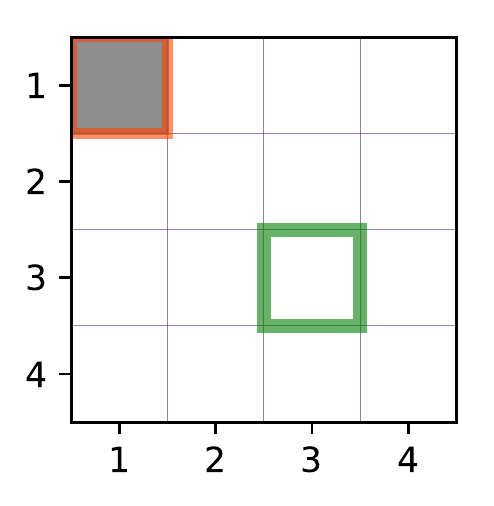}
    \includegraphics[trim={8pt 5pt 8pt 5pt},height=.11\textwidth]{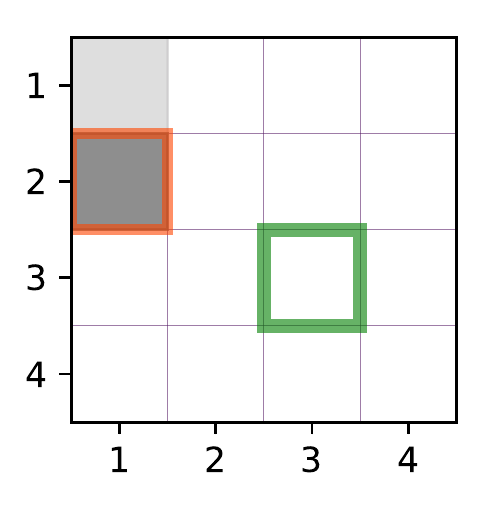}
    \includegraphics[trim={8pt 5pt 8pt 5pt},height=.11\textwidth]{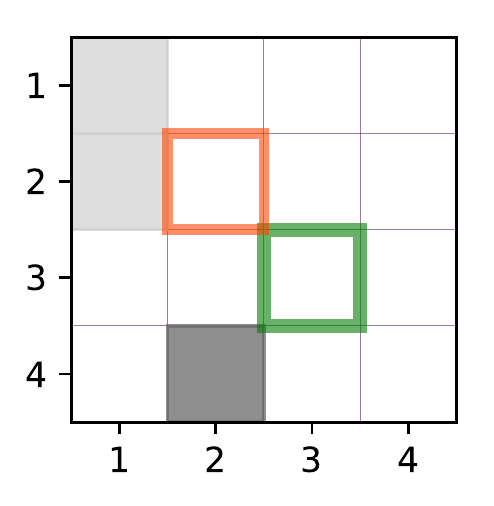}
    \includegraphics[trim={8pt 5pt 8pt 5pt},height=.11\textwidth]{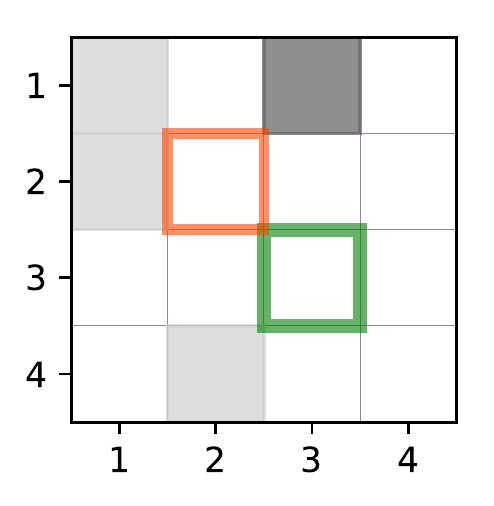}
    \includegraphics[trim={8pt 5pt 8pt 5pt},height=.11\textwidth]{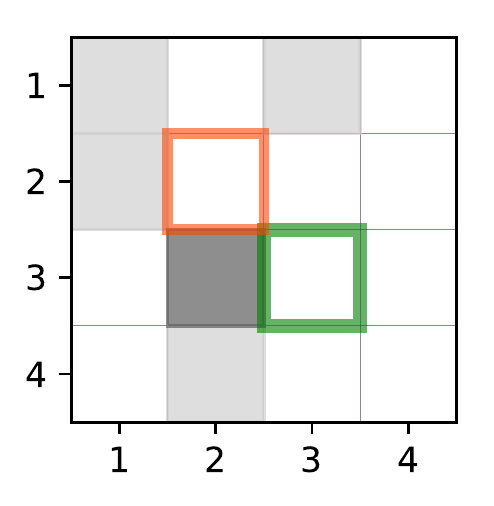}
    \includegraphics[trim={8pt 5pt 8pt 5pt},height=.11\textwidth]{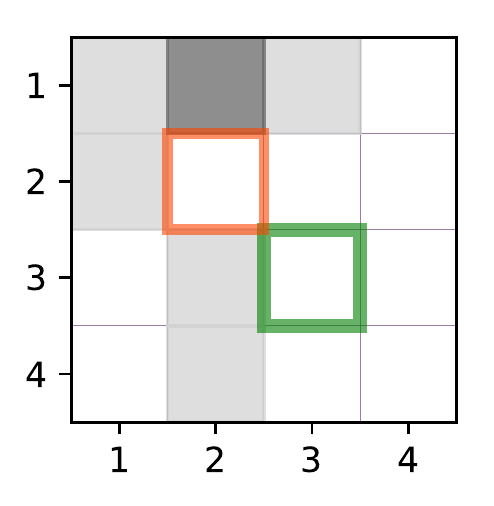}
    \includegraphics[trim={8pt 5pt 8pt 5pt},height=.11\textwidth]{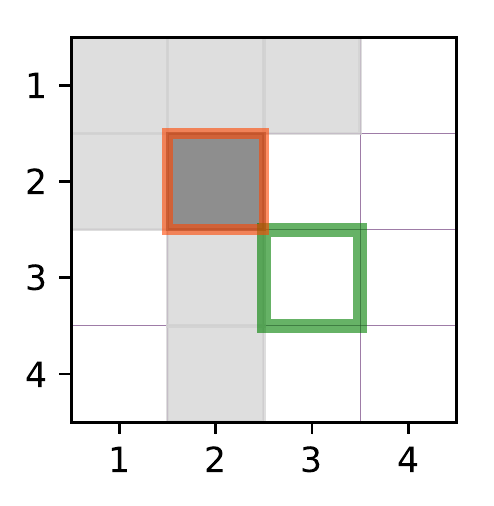}
    \includegraphics[trim={8pt 5pt 8pt 5pt},height=.11\textwidth]{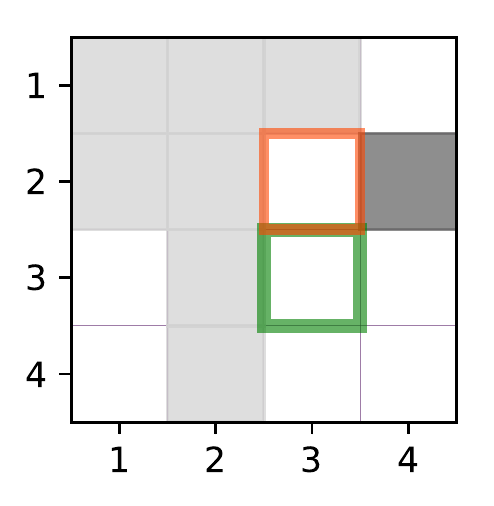}
    \includegraphics[trim={8pt 5pt 8pt 5pt},height=.11\textwidth]{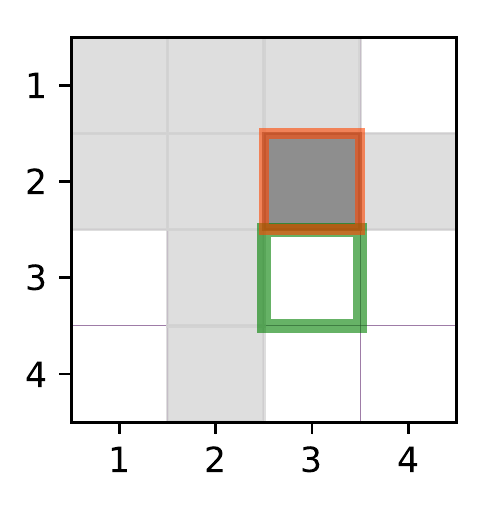}
    \caption{Myopic}
    \label{fig:app:lattice:myopic}
  \end{subfigure}
  \vspace{.005cm}
  
  \begin{subfigure}[b]{1.0\textwidth}
    \includegraphics[trim={8pt 5pt 8pt 5pt},height=.11\textwidth]{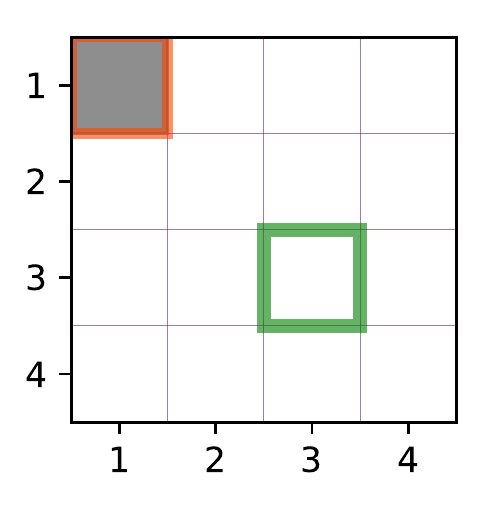}
    \includegraphics[trim={8pt 5pt 8pt 5pt},height=.11\textwidth]{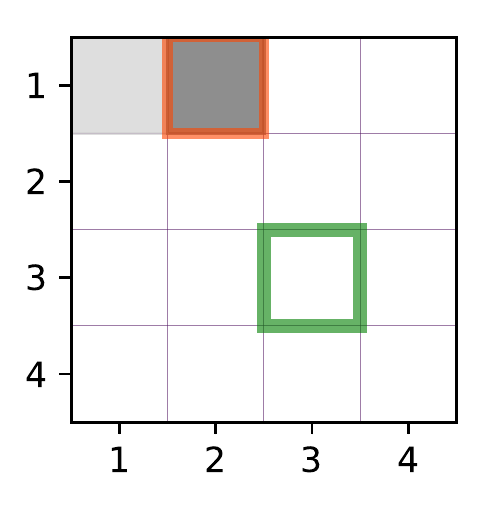}
    \includegraphics[trim={8pt 5pt 8pt 5pt},height=.11\textwidth]{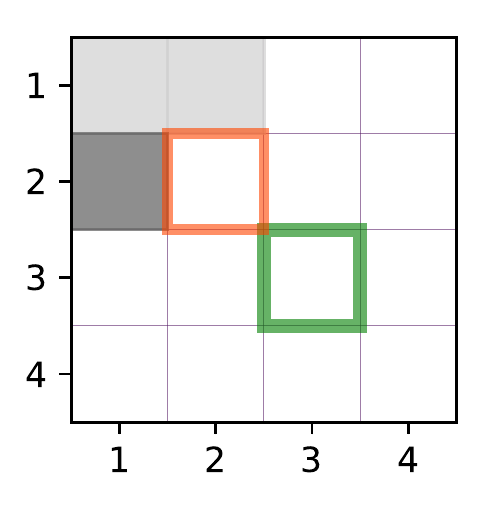}
    \includegraphics[trim={8pt 5pt 8pt 5pt},height=.11\textwidth]{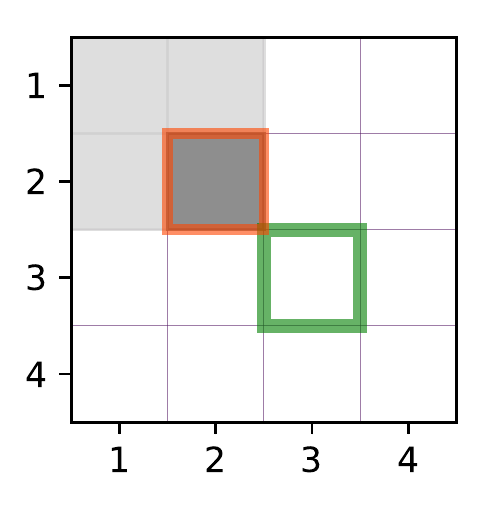}
    \includegraphics[trim={8pt 5pt 8pt 5pt},height=.11\textwidth]{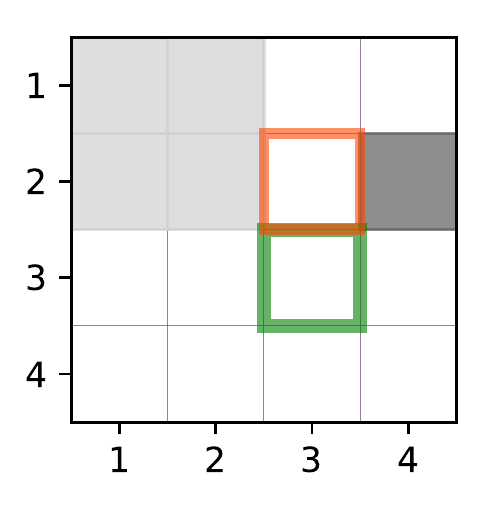}
    \includegraphics[trim={8pt 5pt 8pt 5pt},height=.11\textwidth]{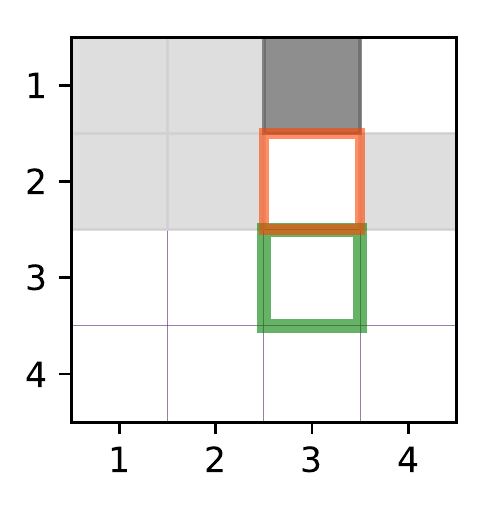}
    \includegraphics[trim={8pt 5pt 8pt 5pt},height=.11\textwidth]{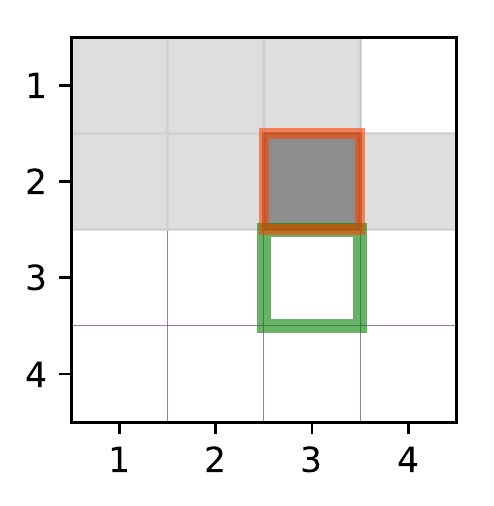}
    \caption{Ada-L}
    \label{fig:app:lattice:adal}
  \end{subfigure}
  \caption{Teaching sequences generated by Myopic and Ada-L on a $4\times 4$ lattice, with $h_0=(1,1), h^*=(3,3)$. The learner's current hypothesis is marked by orange, and the target is marked by green. The dark gray square represents the teaching example at the current time step, while light gray squares represent the previous teaching examples.}
  \label{fig:app:lattice_2}
\end{figure*}

\begin{figure*}[!t]
  \centering
  \begin{subfigure}[b]{1\textwidth}
    \includegraphics[trim={8pt 5pt 8pt 5pt},height=.11\textwidth]{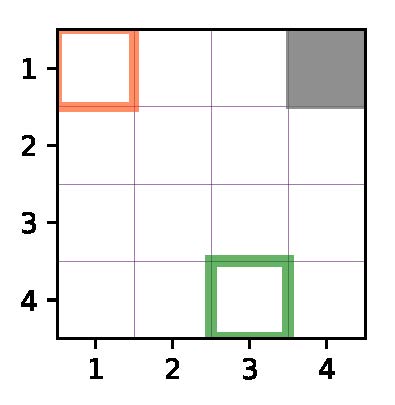}
    \includegraphics[trim={8pt 5pt 8pt 5pt},height=.11\textwidth]{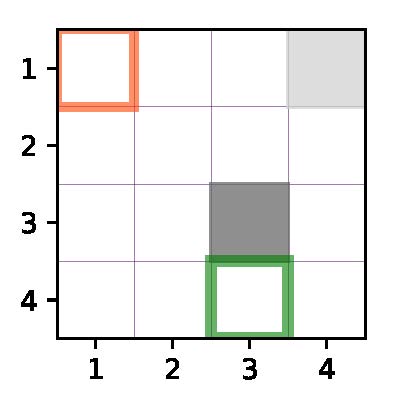}
    \includegraphics[trim={8pt 5pt 8pt 5pt},height=.11\textwidth]{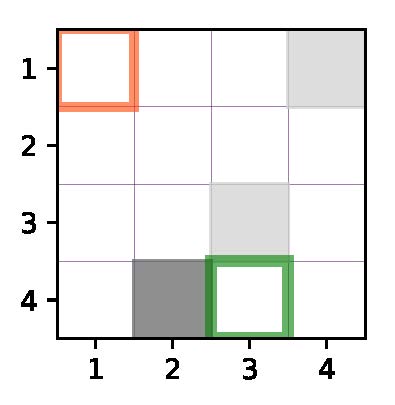}
    \includegraphics[trim={8pt 5pt 8pt 5pt},height=.11\textwidth]{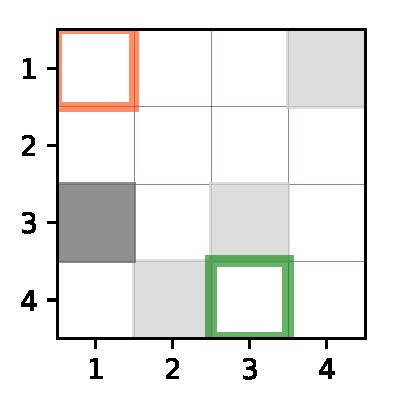}
    \includegraphics[trim={8pt 5pt 8pt 5pt},height=.11\textwidth]{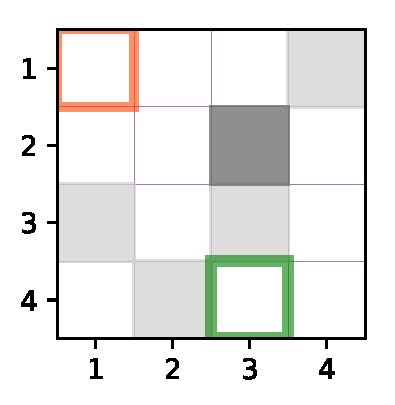}
    \includegraphics[trim={8pt 5pt 8pt 5pt},height=.11\textwidth]{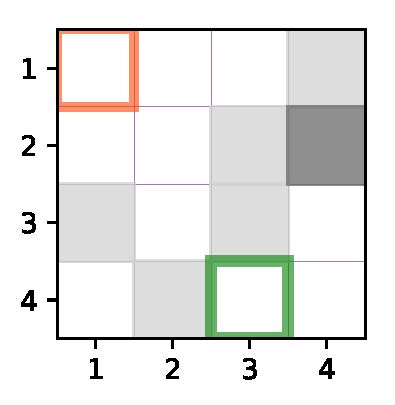}
    \includegraphics[trim={8pt 5pt 8pt 5pt},height=.11\textwidth]{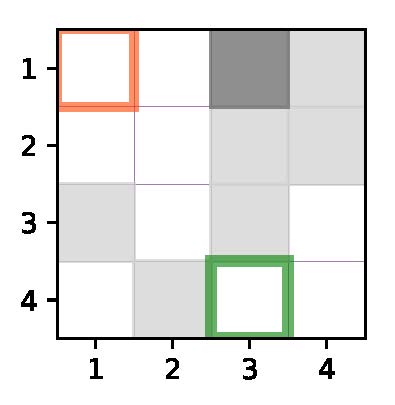}
    \includegraphics[trim={8pt 5pt 8pt 5pt},height=.11\textwidth]{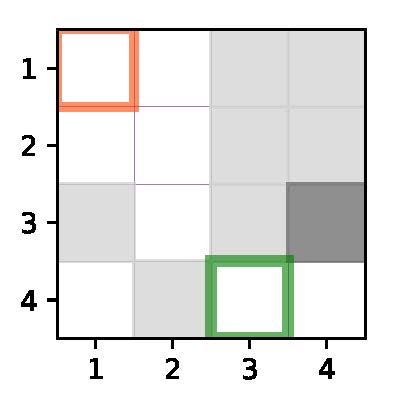}
    \includegraphics[trim={8pt 5pt 8pt 5pt},height=.11\textwidth]{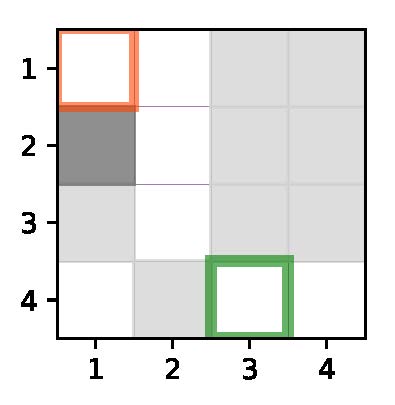}
    \caption{Myopic}
    \label{fig:app:lattice:myopic}
  \end{subfigure}
  \begin{subfigure}[b]{1\textwidth}
    \includegraphics[trim={8pt 5pt 8pt 5pt},height=.11\textwidth]{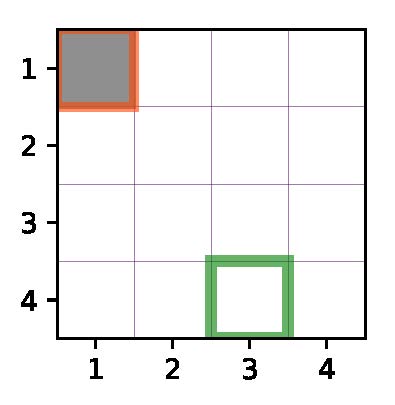}
    \includegraphics[trim={8pt 5pt 8pt 5pt},height=.11\textwidth]{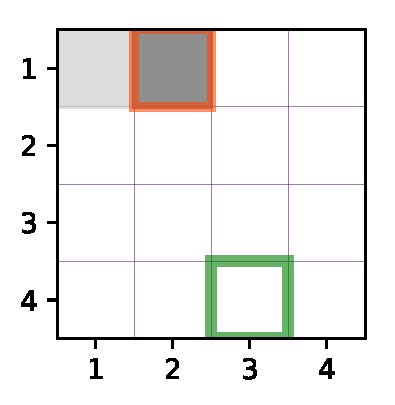}
    \includegraphics[trim={8pt 5pt 8pt 5pt},height=.11\textwidth]{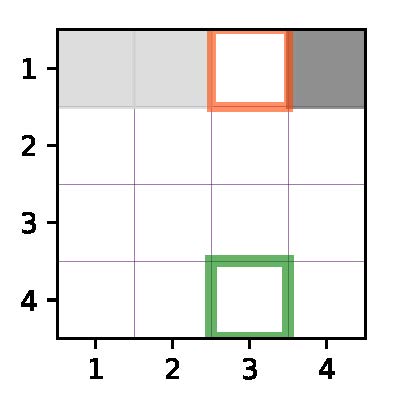}
    \includegraphics[trim={8pt 5pt 8pt 5pt},height=.11\textwidth]{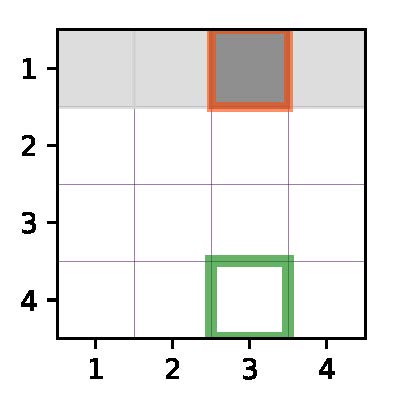}
    \includegraphics[trim={8pt 5pt 8pt 5pt},height=.11\textwidth]{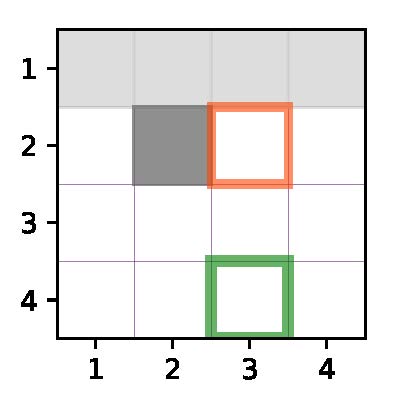}
    \includegraphics[trim={8pt 5pt 8pt 5pt},height=.11\textwidth]{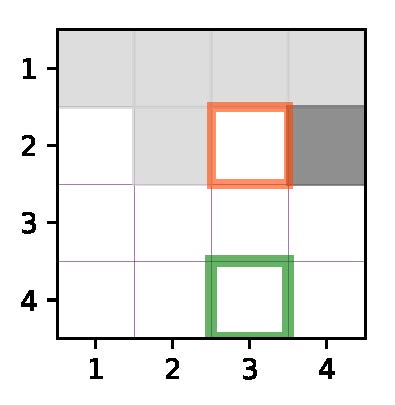}
    \includegraphics[trim={8pt 5pt 8pt 5pt},height=.11\textwidth]{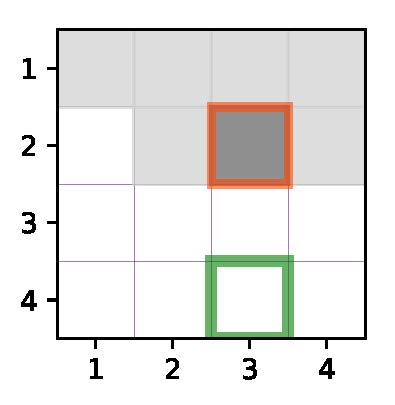}
    \includegraphics[trim={8pt 5pt 8pt 5pt},height=.11\textwidth]{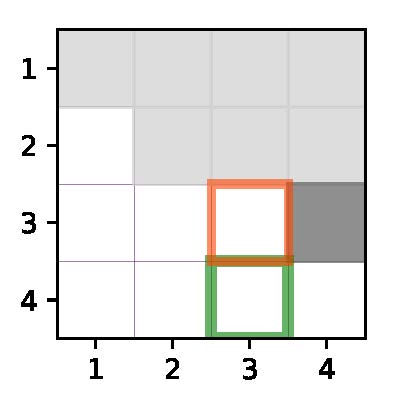}
    \includegraphics[trim={8pt 5pt 8pt 5pt},height=.11\textwidth]{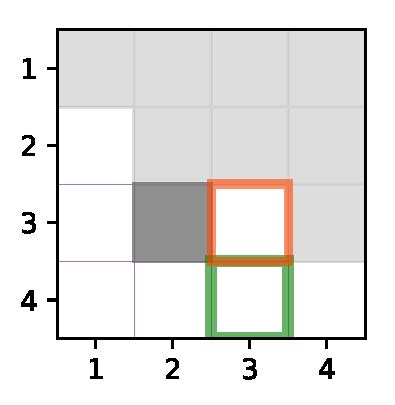}
        \caption{Ada-L}
    \label{fig:app:lattice:adal}
  \end{subfigure}
  \caption{Teaching sequences generated by  Myopic and Ada-L algorithms on a $4\times 4$ lattice, with $h_0=(1,1), h^*=(3,4)$. The learner's initial hypothesis is marked by orange, and the target is marked by green. The dark gray square represents the teaching example at the current time step, while light gray squares represent the previous teaching examples.}
  \label{fig:app:lattice_3}
\end{figure*}

In order to illustrate the proposed framework, we consider a toy scenario, where the teacher aims to teach/steer a human
learner to reach a goal state in a physical environment. Each hypothesis/node corresponds to some
unexplored territory, and there exists an example which flags the territory as explored. The learner
prefers local moves, and if all neighboring territories are explored, the learner jumps to the next
closest one.

The physical environment is characterized by a $4\times 4$ lattice corresponding to $16$ hypotheses. The teacher has $16$ choices of locations on the lattice to show to the student as examples. The student then receives two labels based on its answer $y \in \{-1,1\}$, $-1$ for wrong and $1$ for right.  The preference function $\sigma(h';h)$ is given by the minimum distance between hypotheses described by~$\ell_1(h';h)$.

In this example, we compare two teaching algorithms in the adaptive setting, where the teacher observes the learner's hypothesis at each iteration. The Myopic algorithm is a greedy approach which, at each iteration, picks the teaching example such that after observing the label, the worst-case rank of the target hypothesis in the learner's resulting version space is the smallest. The Ada-L algorithm aims to teach the learner some intermediate hypothesis at each iteration, i.e., it aims to direct the learner to transit to a hypothesis that is ``closer'' to the target hypothesis. For more details of the algorithms please refer to~\cite{chen18adaptive}.

Each algorithm provides a set of policies for which we seek to find the minimum number of trials $\tau$ such that the belief in the target hypothesis $h^*$ is greater than a teaching performance~$\lambda$
$$
b_{\tau} (h^*) \ge \lambda.
$$
To this end, we minimize the number of trials $\tau$ such that (43)-(45) are satisfied. We start by a large number of trials ($16$ in this case) and decrease it until no barrier certificate can be found to verify the teaching performance.  We fix the degree of variables $B_i$, $p^{l_1}_i$, $p^{l_2}_i$, $p^{l_3}_i$, and   $p^f_i \in {\Sigma}[b]$, $i \in \{1,2,\ldots,N\}$ in Corollary~5 to~$2$ and search for the certificates. In order to check the SOS conditions, we use diagonally-dominant-SOS (DSOS) relaxations of the SOS programs implemented through the SPOTless tool~\cite{spot} (for more details see~\cite{ahmadi2017dsos,8263706}).

The results on finding the minimum number of trials $\tau$ for which the teaching performance is satisfied were as follows. 
\subsubsection{$h_0=(1,1)$ and $h^*=(3,3)$}  For the Myopic algorithm, for $\lambda=0.8$, we could  find a certificate for $t=11$. Changing the the teaching performance to $\lambda=0.6$ yielded certificates for only $\tau=9$. In contrast, for the Ada-L algorithm, we could obtain $\tau=6$ assuring teaching performance $\lambda=0.8$ and $\tau=7$ guaranteeing teaching performance $\lambda=0.9$.

\subsubsection{$h_0=(1,1)$ and $h^*=(3,4)$}  For the Myopic algorithm, we could not find any certificate for $\lambda=0.8$. Changing the the teaching performance to $\lambda=0.55$ yielded certificates for only $\tau=15$. On the other hand, for the Ada-L algorithm, we obtained $\tau=9$ assuring teaching performance $\lambda=0.8$ and $\tau=10$ assuring teaching performance $\lambda=0.9$.

%\subsubsection{$h_0=(1,1)$ and $h^*=(3,4)$}  For the Myopic algorithm, we could not find any certificate for $\lambda=0.8$. Changing the the teaching performance to $\lambda=0.55$ yielded certificates for only $\tau=15$. On the other hand, for the Ada-L algorithm, we obtained $\tau=9$ assuring teaching performance $\lambda=0.8$ and $\tau=10$ assuring teaching performance $\lambda=0.9$.

%For example, in the first column of \figref{fig:app:lattice_1}, Myopic (Algorithm 1) picks any teaching example, because the worst-case rank of the target hypothesis at $(4,4)$ is the same for any teaching example. The Ada-L algorithm, on the other hand, tries to direct the learner to make a local move from $(1,1)$ to $(1,2)$ or $(2,1)$, which are closer to the target at $(4,4)$. 
%\yuxin{}

The results can also be corroborated from simulations. As can be see in Figures~\ref{fig:app:lattice_2} and~\ref{fig:app:lattice_3}, the Myopic algorithm  perform poorly on simple teaching tasks as compared to the Ada-L algorithm.

\section{CONCLUSIONS AND FUTURE WORK} \label{sec:conclusions}

We studied the reachability, optimality, and safety verification problems of POMDPs using methods from control theory, such as Lyapunov functions and barrier certificates. We showed how in several cases our calculations can be decomposed and computationally implemented as a set of polynomial optimization problems. We used two examples to illustrate our proposed methodology.

Future work will include a synthesis algorithm for POMDPs. To this end, the literature on designing switching sequences of hybrid systems seems relevant~\cite{stellato2017second}. Such synthesis techniques are aimed at providing guarantees for reachability, optimality, and safety in POMDPs without explicitly or approximately solving them.

Our method was implemented as the solution to a set of sum-of-squares programs. The incorporation of special sparsity patterns in sum-of-squares programs has been shown to be very effective in reducing the computational burden of solving them. We will explore the possibility of using recent results on chordal sparsity in sum-of-squares programs~\cite{zheng2018sparse,ahmadi2017improving} in our computational formulation. In another vein, we can circumvent the computational bottleneck of barrier certificates by considering instead barrier functions~\cite{ames2016control}. This is simply because a closed-form expression for a barrier function can be derived from the definition of the safe set (see recent results on using barrier functions for safe policy synthesis of multi-agent POMDPs~\cite{ahmadi2019safe}). 

%%%%%%%%%%%%%%%%%%%%%%%%%%%%%%%%%%%%%%%%%%%%%%%%%%%%%%%%%%%%%%%%%%%%%%%%%%%%%%%%
\section*{ACKNOWLEDGMENTS}
M. Ahmadi appreciates the stimulating discussions with and the teaching sequences provided  by  Dr.~Yuxin Chen  at the California Institute of Technology.

\appendix
\label{appendix}
\subsection{Sum-of-Squares Polynomials} \label{app:SOS}
A polynomial $p(x)$ is a sum-of-squares polynomial if $\exists p_i(x) \in P[x]$, $i \in \{1, \ldots, n_d\}$ such that $p(x) = \sum_i p_i^2(x)$. Hence $p(x)$ is clearly non-negative. A set of polynomials $p_i$ is called \emph{SOS decomposition} of $p(x)$. The converse does not hold in general, that is, there exist non-negative polynomials which do not have an SOS decomposition~\cite{Par00}.  The computation of SOS decompositions, can be cast as an SDP (see~\cite{Par00,choi1995sums,CTVG99}). The Theorem below proves that, in sets satisfying a property stronger than compactness, any positive polynomial can be expressed as a combination of sum-of-squares polynomials and polynomials describing the set.  

For a set of polynomials $\bar{g} = \{g_1(x), \ldots, g_m(x)\}$, $m \in \mathbb{Z}_{\ge1}$, the \emph{quadratic module} generated by $m$ is 
\begin{equation}
M(\bar{g}):= \left\lbrace \sigma_0 +\sum_{i = 1}^{m} \sigma_i g_i | \sigma_i \in \Sigma[x]\right\rbrace.
\end{equation}
A quadratic module $M\in P[x]$ is said \emph{archimedean} if $\exists N \in \mathbb{Z}_{\ge1}$ such that $N - |x|^2 \in M.$ An archimedian set is always compact~\cite{NS08}. At this point, we recall the following result~\cite[Theorem 2.14]{Las09}.

\begin{thm}[Putinar Positivstellensatz]
\label{theorem:Psatz}
Suppose the quadratic module $M(\bar{g})$ is archimedian. Then for every $f \in P[x]$, $$f>0~\forall~x\in \{x | g_1(x)\geq 0, \ldots, g_m(x)\geq 0 \} \Rightarrow f \in (\bar{g}).$$
\end{thm}

The subsequent proposition formalizes the problem of constrained positivity of polynomials which is a direct result of applying Positivstellensatz.
\begin{prop}[\cite{chesi2010lmi}] \label{chesip}
Let $\{a_i\}_{i=1}^k$ and $\{b_i\}_{i=1}^l$ belong to $\mathcal{P}$, then
\begin{eqnarray}
p(x) \ge 0 \quad &\forall x \in \mathbb{R}^n: a_i(x)=0, \, \forall i=1,2,...,k & \nonumber \\
& \text{and} \quad b_j(x) \ge 0, \, \forall j=1,2,...,l&
\end{eqnarray}
is satisfied, if the following holds
\begin{eqnarray} \label{chesieq}
&\exists r_1,r_2,\ldots,r_k \in P[x] \quad \text{and} \quad \exists s_0,s_1,\ldots,s_l \in \Sigma[x] & \nonumber \\
&p = \sum_{i=1}^k r_i a_i +\sum_{i=1}^l s_i b_i +s_0&
\end{eqnarray}
\end{prop}
\begin{prop} \label{spos}
The multivariable polynomial $p(x)$ is strictly positive ($p(x)>0 \quad \forall x \in \mathbb{R}^n$), if there exists a $\lambda > 0$ such that
\begin{equation}
\big( p(x) - \lambda \big) \in \Sigma[x].
\end{equation}
\end{prop}

%%%%%%%%%%%%%%%%%%%%%%%%%%%%%%%%%%%%%%%%%%%%%%%%%%%%%%%%%%%%%%%%%%%%%%%%%%%%%%%%

\bibliography{references}

\begin{thebibliography}{10}
\providecommand{\url}[1]{#1}
\csname url@samestyle\endcsname
\providecommand{\newblock}{\relax}
\providecommand{\bibinfo}[2]{#2}
\providecommand{\BIBentrySTDinterwordspacing}{\spaceskip=0pt\relax}
\providecommand{\BIBentryALTinterwordstretchfactor}{4}
\providecommand{\BIBentryALTinterwordspacing}{\spaceskip=\fontdimen2\font plus
\BIBentryALTinterwordstretchfactor\fontdimen3\font minus
  \fontdimen4\font\relax}
\providecommand{\BIBforeignlanguage}[2]{{%
\expandafter\ifx\csname l@#1\endcsname\relax
\typeout{** WARNING: IEEEtran.bst: No hyphenation pattern has been}%
\typeout{** loaded for the language `#1'. Using the pattern for}%
\typeout{** the default language instead.}%
\else
\language=\csname l@#1\endcsname
\fi
#2}}
\providecommand{\BIBdecl}{\relax}
\BIBdecl

\bibitem{Put94}
M.~L. Puterman, \emph{{{M}arkov} Decision Processes: Discrete Stochastic
  Dynamic Programming}.\hskip 1em plus 0.5em minus 0.4em\relax John Wiley and
  Sons, 1994.

\bibitem{kaelbling1998planning}
L.~P. Kaelbling, M.~L. Littman, and A.~R. Cassandra, ``Planning and acting in
  partially observable stochastic domains,'' \emph{Artificial Intelligence},
  vol. 101, no.~1, pp. 99--134, 1998.

\bibitem{thrun2005probabilistic}
S.~Thrun, W.~Burgard, and D.~Fox, \emph{Probabilistic Robotics}.\hskip 1em plus
  0.5em minus 0.4em\relax The MIT Press, 2005.

\bibitem{WongpiromsarnF12}
T.~Wongpiromsarn and E.~Frazzoli, ``Control of probabilistic systems under
  dynamic, partially known environments with temporal logic specifications,''
  in \emph{CDC}.\hskip 1em plus 0.5em minus 0.4em\relax IEEE, 2012, pp.
  7644--7651.

\bibitem{Russell-AI-Modern}
S.~J. Russell and P.~Norvig, \emph{Artificial Intelligence: A Modern Approach},
  2nd~ed.\hskip 1em plus 0.5em minus 0.4em\relax Pearson Education, 2003.

\bibitem{ShaniPK13}
G.~Shani, J.~Pineau, and R.~Kaplow, ``A survey of point-based {POMDP}
  solvers,'' \emph{Autonomous Agents and Multi-Agent Systems}, vol.~27, no.~1,
  pp. 1--51, 2013.

\bibitem{MadaniHC99}
O.~Madani, S.~Hanks, and A.~Condon, ``On the undecidability of probabilistic
  planning and infinite-horizon partially observable {Markov} decision
  problems,'' in \emph{AAAI}.\hskip 1em plus 0.5em minus 0.4em\relax {AAAI}
  Press, 1999, pp. 541--548.

\bibitem{braziunas2003pomdp}
D.~Braziunas, ``Pomdp solution methods,'' University of Toronto, Tech. Rep.,
  2003.

\bibitem{ChatterjeeCT16}
K.~Chatterjee, M.~Chmel{\'\i}k, and M.~Tracol, ``What is decidable about
  partially observable {Markov} decision processes with {$\omega$}-regular
  objectives,'' \emph{Journal of Computer and System Sciences}, vol.~82, no.~5,
  pp. 878--911, 2016.

\bibitem{Hauskrecht2000}
M.~Hauskrecht, ``Value-function approximations for partially observable markov
  decision processes,'' \emph{J. Artif. Int. Res.}, vol.~13, no.~1, pp. 33--94,
  Aug. 2000.

\bibitem{wu2018privacy}
B.~Wu and H.~Lin, ``Privacy verification and enforcement via belief
  abstraction,'' \emph{IEEE control systems letters}, vol.~2, no.~4, pp.
  815--820, 2018.

\bibitem{Spaan2005}
M.~T.~J. Spaan and N.~Vlassis, ``Perseus: Randomized point-based value
  iteration for {POMDPs},'' \emph{J. Artif. Int. Res.}, vol.~24, no.~1, pp.
  195--220, Aug. 2005.

\bibitem{6284837}
O.~Brock, J.~Trinkle, and F.~Ramos, \emph{{SARSOP}: Efficient Point-Based
  {POMDP} Planning by Approximating Optimally Reachable Belief Spaces}.\hskip
  1em plus 0.5em minus 0.4em\relax MIT Press, 2009, pp. 65--72.

\bibitem{hsu2007accelerating}
D.~Hsu, W.~S. Lee, and N.~Rong, ``Accelerating point-based {POMDP} algorithms
  through successive approximations of the optimal reachable space,'' 2007.

\bibitem{pineau2006anytime}
J.~Pineau, G.~Gordon, and S.~Thrun, ``Anytime point-based approximations for
  large {POMDPs},'' \emph{Journal of Artificial Intelligence Research},
  vol.~27, pp. 335--380, 2006.

\bibitem{lee2008makes}
W.~S. Lee, N.~Rong, and D.~Hsu, ``What makes some {POMDP} problems easy to
  approximate?'' in \emph{Advances in neural information processing systems},
  2008, pp. 689--696.

\bibitem{kurniawati2008sarsop}
H.~Kurniawati, D.~Hsu, and W.~S. Lee, ``{SARSOP}: Efficient point-based {POMDP}
  planning by approximating optimally reachable belief spaces.'' in
  \emph{Robotics: Science and systems}, vol. 2008.\hskip 1em plus 0.5em minus
  0.4em\relax Zurich, Switzerland., 2008.

\bibitem{spaan2005perseus}
M.~Spaan and N.~Vlassis, ``Perseus: Randomized point-based value iteration for
  {POMDPs},'' \emph{Journal of artificial intelligence research}, vol.~24, pp.
  195--220, 2005.

\bibitem{luo2016importance}
Y.~Luo, H.~Bai, D.~Hsu, and W.~S. Lee, ``Importance sampling for online
  planning under uncertainty,'' \emph{The International Journal of Robotics
  Research}, p. 0278364918780322, 2016.

\bibitem{agha2014robust}
A.~Agha-Mohammadi, S.~Agarwal, A.~Mahadevan, S.~Chakravorty, D.~Tomkins,
  J.~Denny, and N.~M. Amato, ``Robust online belief space planning in changing
  environments: Application to physical mobile robots.'' in \emph{ICRA}, 2014,
  pp. 149--156.

\bibitem{prentice2010belief}
S.~Prentice and N.~Roy, ``The belief roadmap: Efficient planning in linear
  {POMDPs} by factoring the covariance,'' in \emph{Robotics Research}.\hskip
  1em plus 0.5em minus 0.4em\relax Springer, 2010, pp. 293--305.

\bibitem{Kochenderfer2013}
M.~J. Kochenderfer and J.~P. Chryssanthacopoulos, ``Collision avoidance using
  partially controlled {M}arkov decision processes,'' in \emph{Agents and
  Artificial Intelligence}, J.~Filipe and A.~Fred, Eds.\hskip 1em plus 0.5em
  minus 0.4em\relax Springer, 2013, vol. 271, pp. 86--100.

\bibitem{smith2004heuristic}
T.~Smith and R.~Simmons, ``Heuristic search value iteration for {POMDPs},'' in
  \emph{Proceedings of the 20th conference on Uncertainty in artificial
  intelligence}.\hskip 1em plus 0.5em minus 0.4em\relax AUAI Press, 2004, pp.
  520--527.

\bibitem{AWLT18}
M.~Ahmadi, B.~Wu, H.~Lin, and U.~Topcu, ``Privacy verification in {POMDPs} via
  barrier certificates,'' in \emph{Decision and Control (CDC), 2018 IEEE 57th
  Annual Conference on,}, 2018.

\bibitem{ACJT18}
M.~Ahmadi, M.~Cubuktepe, N.~Jansen, and U.~Topcu, ``Verification of uncertain
  {POMDPs} using barrier certificates,'' in \emph{56th Annual Allerton
  Conference on Communication, Control, and Computing,}, 2018.

\bibitem{Sondik78}
E.~J. Sondik, ``The optimal control of partially observable {M}arkov processes
  over the infinite horizon: Discounted costs,'' \emph{Operations Research},
  vol.~26, no.~2, pp. 282--304, 1978.

\bibitem{astrom}
K.~J. {Astr\"{o}m}, ``Optimal control of {Markov} decision processes with
  incomplete state estimation,'' \emph{J. Math. Anal. Appl.}, vol.~10, pp.
  174--205, 1965.

\bibitem{goebel2009hybrid}
R.~Goebel, R.~G. Sanfelice, and A.~R. Teel, ``Hybrid dynamical systems,''
  \emph{IEEE Control Systems}, vol.~29, no.~2, pp. 28--93, 2009.

\bibitem{ahmadi2008non}
A.~A. Ahmadi and P.~A. Parrilo, ``Non-monotonic {L}yapunov functions for
  stability of discrete time nonlinear and switched systems,'' in
  \emph{Decision and Control, 2008. CDC 2008. 47th IEEE Conference on}.\hskip
  1em plus 0.5em minus 0.4em\relax IEEE, 2008, pp. 614--621.

\bibitem{KUNDU2017191}
A.~Kundu and D.~Chatterjee, ``On stability of discrete-time switched systems,''
  \emph{Nonlinear Analysis: Hybrid Systems}, vol.~23, pp. 191 -- 210, 2017.

\bibitem{zhang2009exponential}
W.~Zhang, A.~Abate, J.~Hu, and M.~P. Vitus, ``Exponential stabilization of
  discrete-time switched linear systems,'' \emph{Automatica}, vol.~45, no.~11,
  pp. 2526--2536, 2009.

\bibitem{hauskrecht2000value}
M.~Hauskrecht, ``Value-function approximations for partially observable
  {M}arkov decision processes,'' \emph{Journal of artificial intelligence
  research}, vol.~13, pp. 33--94, 2000.

\bibitem{liberzon2003switching}
D.~Liberzon, \emph{Switching in Systems and Control}, ser. Systems \& Control:
  Foundations \& Applications.\hskip 1em plus 0.5em minus 0.4em\relax
  Birkh{\"a}user Boston, 2003.

\bibitem{hespanha2004uniform}
J.~P. Hespanha, ``Uniform stability of switched linear systems: Extensions of
  {LaSalle's} invariance principle,'' \emph{IEEE Transactions on Automatic
  Control}, vol.~49, no.~4, pp. 470--482, 2004.

\bibitem{DBLP:journals/jair/SteinmetzHB16}
M.~Steinmetz, J.~Hoffmann, and O.~Buffet, ``Goal probability analysis in
  probabilistic planning: Exploring and enhancing the state of the art,''
  \emph{J. Artif. Intell. Res.}, vol.~57, pp. 229--271, 2016.

\bibitem{kolobov2012planning}
A.~Kolobov, ``Planning with {M}arkov decision processes: An {AI} perspective,''
  \emph{Synthesis Lectures on Artificial Intelligence and Machine Learning},
  vol.~6, no.~1, pp. 1--210, 2012.

\bibitem{Kat16}
J.-P. Katoen, ``The probabilistic model checking landscape,'' in
  \emph{LICS}.\hskip 1em plus 0.5em minus 0.4em\relax {ACM}, 2016, pp. 31--45.

\bibitem{blanchini2008set}
F.~Blanchini and S.~Miani, \emph{Set-theoretic methods in control}.\hskip 1em
  plus 0.5em minus 0.4em\relax Springer, 2008.

\bibitem{berkenkamp2016safe}
F.~Berkenkamp, R.~Moriconi, A.~P. Schoellig, and A.~Krause, ``Safe learning of
  regions of attraction for uncertain, nonlinear systems with gaussian
  processes,'' in \emph{Decision and Control (CDC), 2016 IEEE 55th Conference
  on}.\hskip 1em plus 0.5em minus 0.4em\relax IEEE, 2016, pp. 4661--4666.

\bibitem{legat2018computing}
B.~Legat, R.~Jungers, and P.~Parrilo, ``Computing controlled invariant sets for
  hybrid systems with applications to model-predictive control,'' \emph{arXiv
  preprint arXiv:1802.04522}, 2018.

\bibitem{burchardt2007estimating}
H.~Burchardt and S.~Ratschan, ``Estimating the region of attraction of ordinary
  differential equations by quantified constraint solving,'' in
  \emph{Proceedings Of The 3rd WSEAS International Conference On Dynamical
  Systems And Control}, vol. 241.\hskip 1em plus 0.5em minus 0.4em\relax
  Citeseer, 2007.

\bibitem{Gurriet2018}
T.~Gurriet, A.~Singletary, J.~Reher, L.~Ciarletta, E.~Feron, and A.~Ames,
  ``Towards a framework for realizable safety critical control through active
  set invariance,'' in \emph{Proceedings of the 9th ACM/IEEE International
  Conference on Cyber-Physical Systems}, ser. ICCPS '18.\hskip 1em plus 0.5em
  minus 0.4em\relax Piscataway, NJ, USA: IEEE Press, 2018, pp. 98--106.

\bibitem{milias2014optimization}
A.~Milias-Argeitis and M.~Khammash, ``Optimization-based lyapunov function
  construction for continuous-time markov chains with affine transition
  rates,'' in \emph{53rd IEEE Conference on Decision and Control}.\hskip 1em
  plus 0.5em minus 0.4em\relax IEEE, 2014, pp. 4617--4622.

\bibitem{7171125}
M.~Ahmadi, G.~Valmorbida, and A.~Papachristodoulou, ``Barrier functionals for
  output functional estimation of {PDEs},'' in \emph{2015 American Control
  Conference (ACC)}, 2015, pp. 2594--2599.

\bibitem{AHMADI201733}
------, ``Safety verification for distributed parameter systems using barrier
  functionals,'' \emph{Systems \& Control Letters}, vol. 108, pp. 33 -- 39,
  2017.

\bibitem{8264626}
M.~Ahmadi, A.~Israel, and U.~Topcu, ``Safety assessment based on
  physically-viable data-driven models,'' in \emph{2017 IEEE 56th Annual
  Conference on Decision and Control (CDC)}, Dec 2017, pp. 6409--6414.

\bibitem{2018arXiv180104072A}
M.~{Ahmadi}, A.~{Israel}, and U.~{Topcu}, ``{Controller Synthesis for Safety of
  Physically-Viable Data-Driven Models},'' \emph{ArXiv e-prints}, Jan. 2018.

\bibitem{packard2010help}
A.~Packard, U.~Topcu, P.~Seiler, and G.~Balas, ``Help on sos,'' \emph{IEEE
  Control Systems}, vol.~30, no.~4, pp. 18--23, 2010.

\bibitem{Stone37}
M.~H. Stone, ``Applications of the theory of {B}oolean rings to general
  topology,'' \emph{Transactions of the American Mathematical Society},
  vol.~41, no.~3, pp. 375--481, 1937.

\bibitem{chen18adaptive}
Y.~Chen, A.~Singla, O.~M. Aodha, P.~Perona, and Y.~Yue, ``Understanding the
  role of adaptivity in machine teaching: The case of version space learners,''
  in \emph{Proc. Conference on Neural Information Processing Systems (NIPS)},
  December 2018.

\bibitem{krishnamurthy2018multiple}
V.~Krishnamurthy, A.~Aprem, and S.~Bhatt, ``Multiple stopping time pomdps:
  Structural results \& application in interactive advertising on social
  media,'' \emph{Automatica}, vol.~95, pp. 385--398, 2018.

\bibitem{yadati2014cavva}
K.~Yadati, H.~Katti, and M.~Kankanhalli, ``Cavva: Computational affective
  video-in-video advertising,'' \emph{IEEE Transactions on Multimedia},
  vol.~16, no.~1, pp. 15--23, 2014.

\bibitem{lehmann2012models}
J.~Lehmann, M.~Lalmas, E.~Yom-Tov, and G.~Dupret, ``Models of user
  engagement,'' in \emph{International Conference on User Modeling, Adaptation,
  and Personalization}.\hskip 1em plus 0.5em minus 0.4em\relax Springer, 2012,
  pp. 164--175.

\bibitem{simard2017machine}
P.~Y. Simard, S.~Amershi, D.~M. Chickering, A.~E. Pelton, S.~Ghorashi, C.~Meek,
  G.~Ramos, J.~Suh, J.~Verwey, M.~Wang \emph{et~al.}, ``Machine teaching: A new
  paradigm for building machine learning systems,'' \emph{arXiv preprint
  arXiv:1707.06742}, 2017.

\bibitem{zhu2015machine}
X.~Zhu, ``Machine teaching: An inverse problem to machine learning and an
  approach toward optimal education.'' in \emph{AAAI}, 2015, pp. 4083--4087.

\bibitem{AWCYT19}
M.~Ahmadi, B.~Wu, Y.~Chen, Y.~Yue, and U.~Topcu, ``Barrier certificates for
  assured machine teaching,'' in \emph{2019 American Control Conference,},
  2019.

\bibitem{spot}
\BIBentryALTinterwordspacing
A.~Megretski, ``Systems polynomial optimization tools {(SPOT)},'' 2010.
  [Online]. Available:
  \url{https://github.com/anirudhamajumdar/spotless/tree/spotless\_isos}
\BIBentrySTDinterwordspacing

\bibitem{ahmadi2017dsos}
A.~A. Ahmadi and A.~Majumdar, ``{DSOS and SDSOS} optimization: more tractable
  alternatives to sum of squares and semidefinite optimization,'' \emph{arXiv
  preprint arXiv:1706.02586}, 2017.

\bibitem{8263706}
A.~A. Ahmadi, G.~Hall, A.~Papachristodoulou, J.~Saunderson, and Y.~Zheng,
  ``Improving efficiency and scalability of sum of squares optimization: Recent
  advances and limitations,'' in \emph{2017 IEEE 56th Annual Conference on
  Decision and Control (CDC)}, Dec 2017, pp. 453--462.

\bibitem{stellato2017second}
B.~Stellato, S.~Ober-Bl{\"o}baum, and P.~J. Goulart, ``Second-order switching
  time optimization for switched dynamical systems,'' \emph{IEEE Transactions
  on Automatic Control}, vol.~62, no.~10, pp. 5407--5414, 2017.

\bibitem{zheng2018sparse}
Y.~Zheng, G.~Fantuzzi, and A.~Papachristodoulou, ``{Sparse sum-of-squares (SOS)
  optimization: A bridge between DSOS/SDSOS and SOS optimization for sparse
  polynomials},'' \emph{arXiv preprint arXiv:1807.05463}, 2018.

\bibitem{ahmadi2017improving}
A.~A. Ahmadi, G.~Hall, A.~Papachristodoulou, J.~Saunderson, and Y.~Zheng,
  ``Improving efficiency and scalability of sum of squares optimization: Recent
  advances and limitations,'' in \emph{Decision and Control (CDC), 2017 IEEE
  56th Annual Conference on}.\hskip 1em plus 0.5em minus 0.4em\relax IEEE,
  2017, pp. 453--462.

\bibitem{ames2016control}
A.~D. Ames, X.~Xu, J.~W. Grizzle, and P.~Tabuada, ``Control barrier function
  based quadratic programs for safety critical systems,'' \emph{IEEE
  Transactions on Automatic Control}, 2016.

\bibitem{ahmadi2019safe}
M.~Ahmadi, J.~W. Singletary, A.and~Burdick, and A.~D. Ames, ``Safe policy
  synthesis in multi-agent {POMDPs} via discrete-time barrier functions,''
  \emph{arXiv preprint arXiv:1903.07823}, 2019.

\bibitem{Par00}
P.~Parrilo, ``Structured semidefinite programs and semialgebraic geometry
  methods in robustness and optimization,'' Ph.D. dissertation, California
  Institute of Technology, 2000.

\bibitem{choi1995sums}
M.~Choi, T.~Y. Lam, and B.~Reznick, ``Sums of squares of real polynomials,'' in
  \emph{Proceedings of Symposia in Pure mathematics}, vol.~58.\hskip 1em plus
  0.5em minus 0.4em\relax American Mathematical Society, 1995, pp. 103--126.

\bibitem{CTVG99}
G.~Chesi, A.~Tesi, A.~Vicino, and R.~Genesio, ``On convexification of some
  minimum distance problems,'' in \emph{5th {E}uropean {C}ontrol {C}onference},
  Karlsruhe, Germany, 1999.

\bibitem{NS08}
M.~Nie, J.and~Schweighofer, ``On the complexity of {P}utinar's
  positivstellensatz,'' \emph{Journal of Complexity}, vol.~23, no.~1, pp.
  135--150, 2007.

\bibitem{Las09}
J.~B. Lasserre, \emph{Moments, Positive Polynomials and Their
  Applications}.\hskip 1em plus 0.5em minus 0.4em\relax Imperial College Press,
  London, 2009.

\bibitem{chesi2010lmi}
G.~Chesi, ``{LMI} techniques for optimization over polynomials in control: a
  survey,'' \emph{IEEE Transactions on Automatic Control}, vol.~55, no.~11, pp.
  2500--2510, 2010.

\end{thebibliography}
\bibliographystyle{IEEEtran}

\end{document}